\documentclass[a4paper, onecolumn, 9pt]{article}
\usepackage[paperwidth=210mm,paperheight=297mm,centering,hmargin=2.3cm,vmargin=2.7cm]{geometry}
\usepackage{marvosym}
\usepackage{amsfonts}
\usepackage{amsthm}
\usepackage{dcolumn}
\usepackage{graphicx}
\usepackage{bm,bbm}
\usepackage{kpfonts}
\usepackage{braket}
\usepackage{color}
\usepackage[T1]{fontenc}
\usepackage[utf8]{inputenc}
\usepackage{authblk}
\newcommand{\cA}{{\mathcal A}}

\newcommand{\cC}{{\mathcal C}}

\newcommand{\cH}{{\mathcal H}}

\newcommand{\cJ}{{\mathcal J}}
\newcommand{\cK}{{\mathcal K}}
\newcommand{\cL}{{\mathcal L}}

\newcommand{\cN}{{\mathcal N}}

\newcommand{\cP}{{\mathcal P}}

\newcommand{\cS}{{\mathcal S}}
\newcommand{\cT}{{\mathcal T}}
\newcommand{\cU}{{\mathcal U}}
\newcommand{\cV}{{\mathcal V}}
\newcommand{\cX}{{\mathcal X}}
\newcommand{\cW}{{\mathcal W}}
\newcommand{\cY}{{\mathcal Y}}

\newcommand{\bbmC}{{\mathbbm C}}

\newcommand{\bbmN}{{\mathbbm N}}
\newcommand{\bbmR}{{\mathbbm R}}

\newcommand{\bbmeins}{{\mathbbm 1}}

\newcommand{\tr}{\mathrm{tr}}

\newcommand{\cl}{\mathrm{cl}}

\newcommand{\conv}{\mathrm{conv}}

\newtheorem{theorem}{Theorem}

\newtheorem{definition}[theorem]{Definition}

\newtheorem{lemma}[theorem]{Lemma}

\newtheorem{proposition}[theorem]{Proposition}
\newtheorem{remark}[theorem]{Remark}

\begin{document}

\title{Universal superposition codes: capacity regions of compound quantum broadcast channel with confidential messages}
\author[1,2]{Holger Boche}
	\author[1]{Gisbert Jan\ss en}
	\author[1]{Sajad Saeedinaeeni}
\affil[1]{{\footnotesize Lehrstuhl f\"ur Theoretische Informationstechnik, Technische Universit\"at M\"unchen, 80290 M\"unchen, Germany}}

\affil[2]{{\footnotesize Munich Center for
   Quantum Science and Technology (MCQST), 80799 M\"unchen, Germany}}
   
   \affil[ ]{{\scriptsize Electronic addresses: \{boche, gisbert.janssen, sajad.saeedinaeeni\}@tum.de}}
\date{}
\maketitle
\section{Abstract}
We derive universal codes for transmission of broadcast and confidential messages over classical-quantum-quantum and fully quantum channels. These codes are robust to channel  uncertainties considered in the compound model. To construct these codes we generalize random codes for transmission of public messages, to derive a universal superposition coding for the compound quantum broadcast channel. As an application, we give a multi-letter characterization of regions corresponding to capacity of the compound quantum broadcast channel for transmitting broadcast and confidential messages simultaneously. This is done for two types of broadcast messages, one called public and the other common. \\

\textit{A version of this paper is published in Journal of Mathematical Physics 61, 042204 (2020). Also, parts of the results will be presented and published at International Symposium on Information Theory (ISIT) 2020}. 

\section{Introduction}
Assuming the state of a channel, connecting two sides of a communication, to be perfectly known by the communicating parties, that would require their complete knowledge of all the relevant physical parameters of the communication channel, is an idealization that often cannot be achieved in real-world applications. The compound channel model, in which the communicating parties only have access to an uncertainty set to which the state of the channel belongs, invokes coding strategies that are robust to such uncertainties. The size of this uncertainty set, depends on the strategy and physical resources used for channel estimation, and under real-life physical communication conditions, will in general be infinite. Relaxing the assumption of the perfectly known channel, requires coding strategies that work for all channels belonging to a set of possibly infinite cardinality and are hence, significantly more sophisticated. A case in point is the coding strategy established by the authors of the current paper in \cite{simcomp} to derive capacity results for simultaneous transmission of classical (public) messages and quantum information over the quantum channel, given that those developed for the perfectly known channel in \cite{devetaknshor} did not provide the structure needed to deal with channel uncertainty. The compound model consists of an indexed set of channels $\{W_{s}\}_{s\in S}$. A channel from this uncertainty set is used in a memoryless fashion for communication, requiring the codes to be reliable for the set $\{W_{s}^{\otimes n}\}_{s\in S}$, where the channel is used  $n\in\mathbb{N}$ times.\\
Information theoretically, the compound model has yielded intriguing properties. One of the interesting information theoretic properties of the compound channel, is that in general, a strong converse cannot be established on the capacity of the compound classical-quantum channel for message transmission, when upper-bounding of the average decoding error is considered. This holds even for finite uncertainty sets \cite{ahls1968,ahls1969,zeroerror}. This observation implies that a second order capacity theorem cannot be developed in this case. Further, calculation of the so called $\epsilon$-capacity of the compound channel under the average error criterion is still an open question. We note however, that determining a second order $\epsilon$-capacity for the compound channel is not possible, due to the observation, that there are examples of the compound channel where the optimistic $\epsilon$-capacity is strictly larger than its pessimistic one (see \cite{BSP} Remark 13). In \cite{ahls2015} Section 3, Ahlswede posed the question of whether or not there exist simple recursive formulas for the $\epsilon$-capacity of the compound channel. This question, being of great practical significance as discussed in the concluding remarks here, was answered negatively by authors of \cite{BSP}. We note also, the importance of codes developed for the compound channel, for another prominent channel model, namely the arbitrarily varying channel (see e.g. \cite{ahls}), where an active jamming party is present. \\
We consider the compound quantum broadcast channel, connecting one sender to two receivers of different permissions or priorities. The channel is used to perform an integrated task, in which a confidential message, kept secret from the third party, is communicated simultaneously with a broadcast message available to both receivers. The requirements on the broadcast  message, determine two communication scenarios. In the first scenario, we consider the case where both receivers are required to decode the broadcast message. We refer to this message as the common message. In the second scenario the decoding condition is relaxed on one of the receivers. That is, the third party, namely the receiver from whom the confidential message is kept secret, may or may not decode the broadcast message, to which, in this scenario, we refer as the public message. \\
 The capacity of the channel for performing such tasks, will include trade-off regions, determining the resourcefulness of the public/common message transmission capacity, for enhancement of confidential message transmission. Information theoretic analysis of these tasks, will naturally be significant when regions beyond those achieved by simple time-sharing between the two tasks are achieved. We first consider the case where the sender is restricted to classical inputs, namely the classical-quantum-quantum (cqq) broadcast model. This model proves useful for obtaining capacity results for the fully quantum broadcast model, where this restriction is lifted. \\
 The classical counterparts of our results were given in \cite{shaefer}. Therein, the authors first derive robust codes for the bidirectional channel, in which both receivers are meant to decode the message. This common message will then piggyback a public message decoded by Bob. The privacy amplification strategies are then applied on part of the public codes to obtain information theoretic security via equivocation. We follow a similar approach in the context of quantum information theory. We obtain codes for the bidirectional channel (broadcast channel with no security requirement) by generalizing the random codes from \cite{mosonyi}. Our generalization of these results (see Appendix \ref{mosonyiapp}), yields a universal superposition coding for cq channels. Our input  structure allows us to use privacy amplification arguments (\cite{minglai}) on part of the codebook to achieve the desired secrecy rates. \\
The quantum broadcast model in which the channel is assumed to perfectly known by communicating parties was considered in \cite{wilde09,wilde12}, with and without a pre-shared secret key respectively. Therein, the authors have established a dynamic capacity trade off region using a coding strategy that is channel-dependent. We use a different strategy in which establish universal superposition codes for the compound bidirectional channel, exploiting properties of Renyi entropies. \\
Another regime in which the quantum broadcast model with confidential messages has been studied, is the one-shot (single serving) model. A one-shot dynamic capacity theorem was derived for regions corresponding to tasks of common, public and private message transmission over the quantum channel in \cite{salek}. It would be interesting to see if the coding strategies used therein, derived from position based decoding (see \cite{anshu17, anshu18, anshu19}), can be used to design codes for the compound channel model.\\
In the first section following this introduction, we introduce the notation used in this work. Precise definitions of channel models, codes and rate regions along with our main results for the cqq model are given in Section \ref{basicdefsec}. We prove the direct part of our capacity results for the cqq model in Section \ref{sect:coding_for_broadcast}, that is followed by the proof of converse in Section \ref{sect:outer_bounds}. The security criterion that we impose on the confidential message, is the mutual information between Alice and Eve to be arbitrarily small for large numbers of channel uses. As the common or indeed the public messages are available to Eve, we require the mentioned mutual information to be conditioned on the broadcast message. Proving the existence of capacity achieving codes is done in two steps. First we consider the case where there is no security criterion placed on the messages sent to Bob and Eve. In this case, we have a bidirectional channel, where Alice, is sending a message to be decoded by Bob and potentially by Eve (weather Eve decodes this message depends on which scenario is considered, determining in turn our labeling of it as common or public). Conditioned on this message (the corresponding codewords are distributed according to a certain structure), Alice is simultaneously transmitting a second type of message, that is decoded by Bob. The random coding that makes precisely this task possible, is given by Lemma \ref{mosonyivariation}, which is our universal superposition coding result. Application of this lemma gives us the desired bidirectional codes in forms of Lemma \ref{lemmanonsecurebcc} (where the conditioning message is common) and \ref{lemmanonsecuretpc} (where the conditioning message is public). In the second step, the second type of message described above, is used for privacy amplification. We give the code definitions and capacity results for the fully quantum channel independently in Section \ref{sect:full_quantum}. Finally, we discuss some further connections of our results with other approaches in quantum information theory in Section \ref{concludingremarks}, along with other concluding remarks and directions for future research.
	\section{Notation and conventions}\label{sect:notation}
	All Hilbert spaces are assumed to have finite dimensions and are over the field $\mathbb{C}$. All alphabets are also assumed to have finite dimensions.The set of linear operators from Hilbert space $\cH$ to itself is denoted by $\cL(\cH)$. We denote the set of states by $\cS(\cH):=\{\rho\in\cL(\cH):\rho\geq 0,\text{tr}(\rho)=1\}$. Pure states are given by projections onto one-dimensional subspaces. Given a unit vector $x\in\cH$, the corresponding pure state will be written as $\ket{x}\bra{x}$. The set of probability distributions on the finite alphabet $\cX$ of cardinality $|\cX|$, will be denoted by $\cP(\cX)$. For $n\in\mathbb{N}$, we define $\cX^{n}:=\big\{(x_{1},\dots,x_{n}):x_{i}\in\cX, \forall i\in\{1,\dots,n\}\big\}$. The sequence $\mathbf{x}$ will denote elements of $\cX^{n}$. Also, we use bold letters to denote vectors (sequences with more that one element). The probability distribution $p^{\otimes n}\in\cP(\cX^{n})$ will be given by $n$-fold product of $p\in\cP(\cX)$, namely $p^{\otimes n}(\mathbf{x})=p(x_{1})\dots p(x_{n})$ with $\mathbf{x}=(x_{1},\dots,x_{n})$. For any number $M\in\mathbb{N}$, we use $[M]:=\{1,\dots,M\}$. \\
	The classical quantum (cq) channel $W:\cX\to\cS(\cH)$, is a completely positive trace preserving map from alphabet $\cX$ to the set of states on Hilbert space $\cH$. We denote the set of all such maps by $CQ(\cX,\cH)$. This set is equipped with the norm $\parallel\cdot\parallel_{CQ}$ defined for $W\in CQ(\cX,\cH)$ by 
	\begin{equation*}
	    \parallel W\parallel_{CQ}:=\max_{x\in\cX}\parallel W(x)\parallel_{1},
	\end{equation*}
	 where $\parallel\cdot\parallel_{1}$ is the trace norm on $\cL(\cH)$. We use the term cqq channel for map $V\in CQ(\cX,\cH_{1}\otimes\cH_{2})$ with two outcomes in two sets of states on two Hilbert spaces.
	 
	 A measurement or a positive operator valued measure (POVM) with $M\in\mathbb{N}$ outcomes on Hilbert space $\cH$, is given by an $M$-tuple $(D_{1},\dots,D_{M}):D_{i}\geq 0$, $\forall i\in[M]$ and $\sum_{i\in[M]}D_{i}=\mathbbm{1}_{\cH}$. With slight abuse of notation, we write $a^c := \bbmeins_{\cH} - a$ for $a\in\cL(\cH)$.\\
	 We use the base two logarithm denoted by $\log$. The von Neumann entropy of a state $\rho \in \cS(\cH)$ is given by
	\begin{equation*}
	S(\rho):=-\text{tr}(\rho\log\rho).
	\end{equation*}

Given the state $\omega_{AB}\in\cS(\cH_{A}\otimes\cH_{B})$, a closely related quantity, namely the mutual information is given by
\begin{equation*}
    I(A;B,\omega):=S(A,\omega)+S(B,\omega)-S(AB,\omega),
\end{equation*}
where $S(\gamma,\omega)$, indicates the von Neumann entropy of the state $\omega_{\gamma}$, the marginal state of $\omega$. Consider the ensemble $\{p(x),\omega_{AB}^{x}\}$ with $\omega_{AB}^{x}\in\cS(\cH_{A}\otimes\cH_{B})$ and $p\in\cP(\cX)$. We can define a classical-quantum (cq) state $\omega_{XAB}\in\cS(\mathbb{C}^{|\cX|}\otimes\cH_{A}\otimes\cH_{B}$, given some ONB $\{e_x\}_{x\in\cX}\in\mathbb{C}^{|\cX|}$ as
\begin{equation}
   \omega_{XAB}:=\sum_{x\in\cX}p(x)\ket{e_{x}}\bra{e_{x}}^{X}\otimes\omega_{AB}^{x}
\end{equation}
Note that we have used the suffix $X$ to label the Hilbert space corresponding to alphabet $\cX$. The conditional mutual information is then defined by
\begin{equation}\label{conditionalmutual}
    I(A;B|X,\omega_{XB}):=\sum_{x\in\cX}I(A;B,\omega_{AB}^{x}).
\end{equation}
 Throughout this work we have made use of finite nets, to approximate arbitrary compound quantum channels using ones with finite uncertainty sets. Relevant definition and statements on nets, are presented in Appendix \ref{nets} by Definition \ref{netsdef} and proceeding lemma.\\
	 We use $\epsilon_{n} \to 0$ exponentially as $n\to\infty$  or we say $\epsilon_{n}$ approaches (goes to) zero exponentially, if $-\frac{1}{n}\log\epsilon_{n}$ is a strictly positive constant. We use $\cl(A)$ to denote the closure of set $A$ and $\conv(A)$ to denote its convex hull. 
	 \section{Basic definitions and main results}\label{basicdefsec}
		In this section we state the main results and definitions for the compound classical-quantum-quantum (cqq) broadcast channel. The results and definitions related to the fully quantum broadcast channel are stated in Section \ref{sect:full_quantum}. For finite alphabet $\cX$ and Hilbert spaces $\cH_{B},\cH_{E}$, let  $\cW:=\{W_{s}\}_{s\in S}\subset CQ(\cX,\cH_{B}\otimes\cH_{E})$ be a set of cqq channels. The compound cqq broadcast channel generated by this set is given by family $\{W_{s}^{\otimes n}, s\in S\}_{n=1}^{\infty}$. 
	In other words, using $n$ instances of the compound channel is equivalent to using $n$ instances of one of the channels from the uncertainty set. The users of this channel may or may not have access to the Channel State Information (CSI). In this document, we consider the case where both users only know the uncertainty set, to which the actual channel belongs.
	We consider two closely related communication scenarios of significance, having both appeared in the literature hitherto. 
	\begin{itemize}
	    \item \textbf{Broadcasting Common and Confidential messages (BCC)}, where the compound channel is used $n\in\mathbb{N}$ times by the sender Alice in control of the input of the channel, to send two types of messages $(m_{0},m_{c})$ simultaneously over the channel.
	\begin{itemize}
		\item $m_{0}\in[M_{0,n}]$, called the common message, that has to be reliably decoded by receiver Bob in control of Hilbert space $\cH_{B}$ and Eve in control of Hilbert space $\cH_{E}$.
		\item $m_{c}\in[M_{c,n}]$, called the confidential message, that has to be decoded reliably by Bob while Eve, the wiretapper, is kept ignorant. 
	\end{itemize}
	\item \textbf{Transmitting Public and Confidential messages (TPC)}, where along with the confidential message $m_{c}\in[M_{0,n}]$ and instead of the common message, Alice wishes to send a "public" message $m_{1}\in[M_{1,n}]$, that is reliably decoded by Bob while it may or may not be decoded by Eve. 
	\end{itemize}
We consider the main concepts and results related to each task in the following. We start with the BCC scenario. The precise definition of the BCC codes is given by the following.\\
\begin{definition}[ BCC codes]
	An $(n,M_{0,n},M_{c,n})$ BCC code for $\cW$, is a family $\cC=(E(\cdot|\mathbf{m}),D_{B,\mathbf{m}},D_{E,m_{0}})_{\mathbf{m}\in\mathbf{M}}$ with $\mathbf{M}:=[M_{0,n}]\times[M_{c,n}]$, stochastic encoder $E: \mathbf{M} \to \cP(\cX^{n})$, POVMs $(D_{B,\mathbf{m}})_{\mathbf{m}\in\mathbf{M}}$ on $\cH_{B}^{\otimes n}$ and $(D_{E,m_{0}})_{m_{0}\in[M_{0,n}]}$ on $\cH_{E}^{\otimes n}$. 
\end{definition}
We define the transmission error functions, for any cqq broadcast channel $W:\cX \to \cS(\cH_{B}\otimes\cH_{E})$ and $n\in\mathbb{N}$ by
\begin{itemize}
	\item $\overline{e}_{B}(\cC,W^{\otimes n}):=\frac{1}{|\mathbf{M}|}\sum_{\mathbf{m}\in\mathbf{M}}\sum_{\mathbf{x}\in\cX^{n}}E(\mathbf{x}|\mathbf{m})\tr(D^{c}_{B,\mathbf{m}}W_{B}^{\otimes n}(\mathbf{x}))$ and
	\item $\overline{e}_{E}(\cC,W^{\otimes n}):=\frac{1}{|\mathbf{M}|}\sum_{m\in\mathbf{M}}\sum_{\mathbf{x}\in\cX^{n}}E(\mathbf{x}|\mathbf{m})\tr(D^{c}_{E,m_{0}}W_{E}^{\otimes n}(\mathbf{x}))$,
\end{itemize}
where, $W_{\gamma}, \gamma\in\{B,E\}$ are the marginal channels of $W$. Moreover, we use the security criterion given by
\begin{equation}
I(M_{c};E|M_{0},\sigma_{s,n}),
\end{equation}
where $\sigma_{s,n}$ is the code state defined by
\begin{equation}\label{securitystate}
\sigma_{s,n}:=\frac{1}{|\mathbf{M}|}\sum_{\mathbf{m}\in\mathbf{M}}\ket{\mathbf{m}}\bra{\mathbf{m}}\otimes\sum_{\mathbf{x}\in\cX^{n}}E(\mathbf{x}|m)W_{s}^{\otimes n}(\mathbf{x}), \qquad (s\in S, n\in\mathbb{N}).
\end{equation}
The conditional mutual information should be understood given (\ref{conditionalmutual}) and considering ONBs $\{\ket{m_{i}}\}_{m_i\in [M_{i}]}\in\mathbb{C}^{M_{i}}$ for $i\in\{0,c\}$ and $\ket{\mathbf{m}}:=\ket{m_{0}}\otimes\ket{m_{c}}$. Based on this, we define the following achievable rate pairs.
\begin{definition}(Achievable BCC rate pair)\label{bccratepairdef}
	A pair $(R_{0},R_{c})$ of non-negative numbers is called an achievable BCC rate pair for $\cW$, if for each $\epsilon,\delta>0$, exists an $n_{0}(\epsilon,\delta)\in\mathbb{N}$, such that for all $n>n_{0}$, we find an $(n,M_{0,n},M_{c,n})$ BCC code $\cC=(E(\cdot|\mathbf{m}),D_{B,\mathbf{m}},D_{E,m_{0}})_{\mathbf{m}\in\mathbf{M}}$ such that
	\begin{enumerate}
		\item $\frac{1}{n}\log M_{i,n}\geq R_{i}-\delta $ ($i\in\{0,c\}$),
		\item $\sup_{s\in S}\overline{e}_{\gamma}(\cC,W_{s}^{\otimes n})\leq \epsilon$ ($\gamma\in\{B,E\}$),
		\item $\sup_{s\in S}I(M_{c};E|M_{0},\sigma_{s,n})\leq\epsilon$,
	\end{enumerate}
	are simultaneously fulfilled.
\end{definition}
We define the BCC capacity region of $\cW$ by 
\begin{align}
C_{BCC}[\cW]:=\{(R_{0},R_{c})\in\mathbb{R}_{0}^{+}\times\mathbb{R}_{0}^{+}:(R_{0},R_{c}) \text{ is achievable BCC rate pair for } \cW\}.
\end{align}
To state our theorem, we define the following regions, given finite alphabets $\cU,\mathcal{Y}$ and probability distribution $p=p_{UYX}\in\cP(\cU\times\mathcal{Y}\times\cX^{n})$, with the random variables $U,Y,X$ distributed accordingly. 
\begin{align*}
\hat{C}^{(1)}\big(\cW,p,n\big)&:=\big\{(R_{0},R_{c})\in\mathbb{R}_{0}^{+}\times\mathbb{R}_{0}^{+}: R_{0}\leq \inf_{s\in S}\min\left\{I(U;B,\omega_{s}),I(U;E,\omega_{s})\right\}\wedge\\& R_{c}\leq\inf_{s\in S}I(Y;B|U,\omega_{s})-\sup_{s\in S}I(Y;E|U,\omega_{s})\big\}.
\end{align*}
with
\begin{equation}\label{evaluationstate}
\omega_{s}:=\sum_{(u,y,\mathbf{x})\in\cU\times\mathcal{Y}\times\cX^n}p(u,y,\mathbf{x})\ket{u}\bra{u}\otimes\ket{y}\bra{y}\otimes W_{s}^{\otimes n}(\mathbf{x}).
\end{equation}
We state the following theorem. 
\begin{theorem}\label{bcctheorem}
	Let $\cW:=\{W_{s}\}_{s\in S}\subset CQ(\cX,\cH_{B}\otimes\cH_{E})$ be any compound cqq broadcast channel. It holds 
	\begin{align}\label{capacity:regions:bcc}
	C_{BCC}[\cW]&=\cl \bigg(\bigcup_{l=1}^{\infty}\bigcup_{p}\frac{1}{l}\hat{C}^{(1)}\big(\cW,p,l\big)\bigg),
	\end{align}
	where we have used $\frac{1}{l}A:=\{(\frac{1}{l}x_{1},\frac{1}{l}x_{2}):(x_{1},x_{2})\in A\}$. The second union is taken over all $p_{UYX}\in\cP(\cU\times\mathcal{Y}\times\cX^{l})$ such that random variable $U-Y-X$ form a Markov chain and alphabets $\cU$ and $\mathcal{Y}$ are finite. 
\end{theorem}
\begin{remark}\label{convex:remark}
The set given on the right hand side of (\ref{capacity:regions:bcc}) is convex and hence we do not need further convexification here. This results from time sharing arguments applied on the entropic quantities appearing in (\ref{capacity:regions:bcc}). For a short proof of a similar statement, see \cite{simcomp}.
\end{remark}
We proceed with the TPC scenario. The precise definition of the TPC codes is given in the following.  
\begin{definition}[ TPC codes]
	An $(n,M_{1,n},M_{c,n})$ TPC code for $\cW$, is a family $\cC=(E(\cdot|\mathbf{m}),D_{B,\mathbf{m}})_{\mathbf{m}\in\mathbf{M}}$ with $\mathbf{M}:=[M_{1,n}]\times[M_{c,n}]$, stochastic encoder $E:\mathbf{M}\to \cP(\cX^{n})$ and a POVM $(D_{B,\mathbf{m}})_{\mathbf{m}\in\mathbf{M}}$ on $\cH_{B}^{\otimes n}$. 
\end{definition}
We define the relevant transmission error function, for any cqq broadcast channel $W:\cX\to\cS(\cH_{B}\otimes\cH_{E})$ and $n\in\mathbb{N}$ by
\begin{align*}
	\overline{e}_{B}(\cC,W^{\otimes n}):=\frac{1}{|\mathbf{M}|}\sum_{\mathbf{m}\in\mathbf{M}}\sum_{\mathbf{x}\in\cX^{n}}E(\mathbf{x}|\mathbf{m})\tr(D^{c}_{B,\mathbf{m}}W_{B}^{\otimes n}(\mathbf{x})).
	\end{align*}
 Moreover, we use the security criterion given by
\begin{equation}
I(M_{c};E|M_{1},\sigma_{s,n}),
\end{equation}
where $\sigma_{s,n}$ is the code state defined by
\begin{equation}\label{securitystateTPC}
\sigma_{s,n}:=\frac{1}{|\mathbf{M}|}\sum_{\mathbf{m}\in\mathbf{M}}\ket{\mathbf{m}}\bra{\mathbf{m}}\otimes\sum_{\mathbf{x}\in\cX^{n}}E(\mathbf{x}|\mathbf{m})W_{s}^{\otimes n}(\mathbf{x}).
\end{equation}
Again, we not that the conditional mutual information should be understood given (\ref{conditionalmutual}) and considering ONBs $\{\ket{m_{i}}\}_{m_i\in [M_{i}]}\in\mathbb{C}^{M_{i}}$ for $i\in\{1,c\}$ and $\ket{\mathbf{m}}:=\ket{m_{1}}\otimes\ket{m_{c}}$. Based on this, we define the following achievable rate pairs.
\begin{definition}(Achievable TPC rate pair)
	A pair $(R_{1},R_{c})$ of non-negative numbers is called an achievable TPC rate pair for $\cW$, if for each $\epsilon,\delta>0$, exists an $n_{0}(\epsilon,\delta)\in\mathbb{N}$, such that for all $n>n_{0}$, we find an $(n,M_{1,n},M_{c,n})$ TPC code $\cC=(E(\cdot|\mathbf{m}),D_{B,\mathbf{m}})_{\mathbf{m}\in\mathbf{M}}$ such that
	\begin{enumerate}
		\item $\frac{1}{n}\log M_{i,n}\geq R_{i}-\delta $ ( $i\in\{1,c\}$),
		\item $\sup_{s\in S}\overline{e}_{B}(\cC,W_{s}^{\otimes n})\leq \epsilon$,
		\item $\sup_{s\in S}I(M_{c};E|M_{1},\sigma_{s,n})\leq\epsilon$
	\end{enumerate}
    are simultaneously fulfilled.
\end{definition}
We define the TPC capacity region of $\cW$ by 
\begin{align}
C_{TPC}[\cW]:=\{(R_{1},R_{c})\in\mathbb{R}_{0}^{+}\times\mathbb{R}_{0}^{+}:(R_{1},R_{c}) \text{ is achievable TPC rate for } \cW\}.
\end{align}
To state our theorem, we define the following sub-regions, given finite alphabets $\cV,\mathcal{Y}$ and probability distribution $p=p_{VYX}\in\cP(\cV\times\mathcal{Y}\times\cX^{n})$, with the random variables $V,Y,X$ distributed accordingly. 
\begin{align*}
C^{(1)}\big(\cW,p,n\big)&:=\big\{(R_{1},R_{c})\in\mathbb{R}_{0}^{+}\times\mathbb{R}_{0}^{+}: R_{1}\leq \inf_{s\in S}I(V;B,\omega_{s})\wedge\\& R_{c}\leq\inf_{s\in S}I(Y;B|V,\omega_{s})-\sup_{s\in S}I(Y;E|V,\omega_{s})\big\}.
\end{align*}
with
\begin{equation}\label{evaluationstateTPC}
\omega_{s}:=\sum_{(v,y,\mathbf{x})\in\cV\times\cX\times\cX}p(v,y,\mathbf{x})\ket{v}\bra{v}\otimes\ket{y}\bra{y}\otimes W_{s}^{\otimes n}(\mathbf{x}).
\end{equation}
We can state the following theorem.
\begin{theorem} \label{tpctheorem}
	Let $\cW:=\{W_{s}\}_{s\in S}\subset CQ(\cX,\cH_{B}\otimes\cH_{E})$ be any compound cqq broadcast channel. It holds 
	\begin{align}\label{capacity:regions:tpc}
	C_{TPC}[\cW]&=\cl \bigg(\bigcup_{l=1}^{\infty}\bigcup_{p}\frac{1}{l}C^{(1)}\big(\cW,p,l\big)\bigg).
	\end{align}
	The second union is taken over all $p_{VYX}\in\cP(\cV\times\mathcal{Y}\times\cX^{l})$ such that random variable $V-Y-X$ form a Markov chain and alphabets $\cV$ and $\mathcal{Y}$ are finite. 
\end{theorem}
Again, we note Remark \ref{convex:remark}, regarding convexity of the set on the right hand side of (\ref{capacity:regions:tpc}).
\section{Coding for broadcast channel} \label{sect:coding_for_broadcast}
In this section we present coding strategies for BCC and TPC communication scenarios sufficient to achieve each point in the capacity region. We prove appropriate inner bounds on the capacity regions, namely the direct parts of the main theorems presented in the previous section. Here, we begin by some preliminary results, in the statements of which, we make use of typical sets and projections. The use of these objects are standard in classical as well as quantum information theory. The reader will find detailed explanations in \cite{korner}. We begin this section nevertheless, by introducing these objects. Given two probability distributions $p\in\cP(\bar{\cX})$ and $\forall x\in\bar{\cX}$,  $t(\cdot|x)\in\cP(\bar{\cY})$, $n\in\mathbb{N}$, $\delta>0$, we define the following sets. The set of $\delta$-typical sequences in $\bar{\mathcal{X}}^{n}$, is defined by
\begin{equation}\label{typicalset}
T_{p,\delta}^{n}:=\{\mathbf{x}:\forall x\in\bar{\cX},|\frac{1}{n}N(x|\mathbf{x})-p(x)|\leq\delta\text{  }
\wedge \text{ } q(x)=0\iff N(x|\mathbf{x})=0\}
\end{equation}
with $N(x|\mathbf{x})$, the number of occurrences of letter $x$ in word $\mathbf{x}$. Also, the set of conditionally typical sequences in $\bar{\cY}^{n}$, is given by
\begin{align*}
T_{t,\delta}(\mathbf{x}):&=\{\mathbf{y}\in\bar{\cY}^{n}:\forall x\in\bar{\cX},y\in\bar{\cY}:|\frac{1}{n}N(x,y|\mathbf{x},\mathbf{y})-\frac{1}{n}t(y|x)N(x|\mathbf{x})|\leq\delta\text{ and }\\&
t(y|x)=0\iff N(x,y|\mathbf{x},\mathbf{y})=0 \text{  for  } x\in\bar{\cX},y\in\bar{\cY}\}.
\end{align*}
The pruned distributions associated with $p$ and $t(\cdot|x)$ are given by the following respectively. 
\begin{equation}\label{pruned1}
p'_{n,\delta}(\mathbf{x}):=
\begin{cases}
\frac{p^{\otimes n}(\mathbf{x})}{p^{\otimes n}(T_{p,\delta}^{n})}, & \text{if}\ \mathbf{x}\in T_{p,\delta}^{n}\\
0, & \text{otherwise},
\end{cases}
\end{equation} 
and
 \begin{equation}\label{pruned2}
t_{n,\delta}'(\mathbf{y}|\mathbf{x}):=
\begin{cases}
\frac{t^{\otimes n}(\mathbf{y}|\mathbf{x})}{t^{\otimes n}(T_{t,\delta}(\mathbf{x})|\mathbf{x})}, & \text{if}\ \mathbf{y}\in T_{t,\delta}(\mathbf{x})\\
0, & \text{otherwise}.
\end{cases}
\end{equation}
For the remainder of this section, pruned distributions defined above, will be denoted by primed letters indicating the probability distribution, indexed by the number of available copies of the system. For instance the pruned probability distribution related to $r\in\cP(\cX)$, over $T_{r,\delta}^{n}$ will be denoted by $r'_{n,\delta}$. In (\ref{typicalset}), when $\delta=0$, we have the exact type notified by $T_{p}^{n}$. We also define the set of types by 
\begin{equation}\label{setoftypes}
    \cT(\bar{\cX},n):=\{\lambda\in\cP(\bar{\cX}):T_{\lambda}^{n}\neq\emptyset\}.
\end{equation}

The following lemma contains the properties typical projections, that projection operators assigned to typical sets.
\begin{lemma}\label{typicalsets}
	Let $\lambda\in\cP(\cA)$ with $\lambda(x)>0$ for all $x\in\cA\subset\cX$, $\{\rho_{x}\}_{x\in\cX}\subset\cS(\cK_{A})$ and $\delta>0$. For $\mathbf{x}\in T_{\lambda,\delta}^{n}$ with $\mathbf{x}:=(x_{1},\dots,x_{n})$ and $\rho_{\mathbf{x}}:=\bigotimes_{i=1}^{n}\rho_{x_{i}}$. Define
	\begin{equation*}
	\theta:=\sum_{x\in\cX}\lambda(x)\ket{x}\bra{x}^{X}\otimes\rho_{x}.
	\end{equation*} 
	There exist positive constants $\Upsilon(\delta),\Gamma(\delta)$ and $\Delta(\delta)$ depending on $\delta$ and an orthogonal projector $\Pi_{\rho_{\mathbf{x}},\delta}$ such that
	\begin{enumerate}
		\item $\tr(\rho_{\mathbf{x}}\Pi_{\rho_{\mathbf{x}},\delta})\geq 1-2^{-n\Upsilon(\delta)}$,
		\item $\tr(\Pi_{\rho_{\mathbf{x}},\delta})\leq 2^{n(S(A|X,\theta)+\Delta(\delta))}$,
		\item $\Pi_{\rho_{\mathbf{x}},\delta}\rho_{\mathbf{x}}\Pi_{\rho_{\mathbf{x}},\delta}\leq 2^{-n(S(A|X,\theta)+\Gamma(\delta))}\Pi_{\rho_{\mathbf{x}},\delta}$,\\
		Also, let $W:\cY\to\cS(\cK_{B})$ be a cq channel and $r(\cdot|x)\in\cP(\cY)$, for all $x\in\cX$. Define the state
		\begin{equation*}
		\theta':=\sum_{(x,y)\in\cX\times\cY}\lambda(x)\ket{x}\bra{x}\otimes r(y|x)\ket{y}\bra{y}\otimes W(y).
		\end{equation*}
		For $\mathbf{y}\in T_{r,\delta}(\mathbf{x})$, there exist positive constants $\Upsilon'(\delta), \Delta'(\delta),\Gamma'(\delta)$ and an orthogonal projector $\Pi_{W,\mathbf{x},\delta}(\mathbf{y})$, commuting with $W^{\otimes n}(\mathbf{y})$, satisfying 
		\item $\tr[W^{\otimes n}(\mathbf{y})\Pi_{W,\mathbf{x},\delta}(\mathbf{y})]\geq 1-2^{-n\Upsilon'(\delta)}$,
		\item $\tr[\Pi_{W,\mathbf{x},\delta}(\mathbf{y})]\leq 2^{n(S(B|XY,\theta)+\Delta'(\delta))}$,
		\item $\Pi_{W,\mathbf{x},\delta}(\mathbf{y})W^{\otimes n}(\mathbf{y})\Pi_{W,\mathbf{x},\delta}(\mathbf{y})\leq 2^{-n(S(B|XY,\theta')+\Gamma'(\delta))}\Pi_{W,\mathbf{x},\delta}(\mathbf{y})$.
	\end{enumerate}
	Finally, we have the following total conditional subspace projection. For $\rho_{x}=\sum_{y\in\cY}r(y|x)W(y)$, the projection $\Pi_{W,\mathbf{x},\delta}:=\Pi_{\rho_{\mathbf{x}},\delta}$ with properties 1-3, for $\mathbf{y}\in T_{r,\delta}(\mathbf{x})$ also has the following property.
	\begin{equation}\label{totproj}
	\tr(\Pi_{W,\mathbf{x},\delta}W^{\otimes n}(\mathbf{y}))\geq 1-2^{-n\Upsilon"(\delta)},
	\end{equation}
	for some constant $\Upsilon"(\delta)>0$ depending on $\delta$. 
\end{lemma} 
\begin{proof}
    Properties 1-3 result directly from Lemma 14 \cite{simcomp}. Properties 4-6 and (\ref{totproj}), result from applying the same concatenation arguments as in the proof of Lemma 14 \cite{simcomp}, on inequalities (4)-(7) from \cite{minglaithesis}. 
\end{proof}
A crucial ingredient for the achievablity proofs in this paper is Lemma \ref{mosonyivariation} below. It states existence of certain universal random codes for cq channels given a "typical word".
\begin{lemma}\label{mosonyivariation}
Let $\{W_{s}\}_{s\in S}\subset CQ(\cY,\cK_{B})$ be any set of cq channels, $q\in\cP(\cX)$ and $r(\cdot|x)\in\cP(\cY)$ for each $x\in\cX$.
	For $\delta>0$, there exists $n_{0}\in\mathbb{N}$, such that for $n>n_{0}$, for each $\mathbf{x}\in T_{q,\delta}^{n}$, there exists a map $y:(y_{1},\dots,y_{M})\mapsto(\Lambda_{1}(y)\dots,\Lambda_{M}(y))$, such that $(\Lambda_{m}(y))_{m\in [M]}\subset\cL(\cK_{B}^{\otimes n})$ is a POVM and for any family $Y:=(Y_{1},\dots,Y_{M})$ of random variables, distributed i.i.d according to $r'_{n,\delta}(\cdot|\mathbf{x})$, namely the pruned distribution of $r(\cdot|x)$ (see (\ref{pruned2})),  we have
	\begin{align*}
	\mathbb{E}_{Y}\big[\sup_{s\in S}&\frac{1}{M}\sum_{m\in[M]}\tr(W_{s}^{\otimes n}(Y_{m})\Lambda^{c}_{m}(Y))\big] \ \leq\epsilon_{n}
	\end{align*}
	with $\epsilon_{n}\to0$ exponentially and \begin{equation*}
	    \frac{1}{n}\log M\geq \inf_{s\in S}I(Y;B|X,\sigma_{s})-c\delta,
	\end{equation*}
	with some constant $c>0$ and
	\begin{equation*}
	    \sigma_{s}:=\sum_{x\in\cX}q(x)\ket{x}\bra{x}\otimes\sum_{y\in\cY}r(y|x)\ket{y}\bra{y}\otimes W_{s}(y). 
	\end{equation*}
	    \end{lemma}
	    \begin{proof}
	    We present a full argument in Appendix \ref{mosonyiapp}.
	    \end{proof}
	     
The following statement is an immediate consequence of the above, for the case $|\cX|=1$. We include this statement for clarity of reference later on. 
\begin{lemma}\label{corollary}
	Let $\{W_{s}\}_{s\in S}\subset CQ(\cY,\cK_{B})$ be any set of cq channels and $r\in\cP(\cY)$. For $\delta>0$, there exists $n_{0}$, such that for $n>n_{0}$, there exists a map $y:(y_{1},\dots,y_{M})\mapsto(\Lambda_{1}(y),\dots,\Lambda_{M}(y))$, such that $(\Lambda_{m}(y))_{m\in[M]}$ is a POVM and for any  family $Y:=(Y_{1},\dots,Y_{M})$ of random variables, distributed i.i.d according to $r'_{n,\delta}$, namely the pruned distribution of $r$ (see (\ref{pruned1})), we have
	\begin{align*}
	\mathbb{E}_{Y}\big[\sup_{s\in S}&\frac{1}{M}\sum_{m\in[M]}\tr(W_{s}^{\otimes n}(Y_{m})\Lambda_{m}^{c}(Y))\big]\leq\epsilon_{n},
	\end{align*}
	with $\epsilon_{n}\to 0$ exponentially and 
	\begin{equation*}
	    \frac{1}{n}\log M\geq\inf_{s\in S }I(Y;B,\sigma_{s})-c\delta
	\end{equation*}
	for some constant $c>0$ and 
	\begin{equation*}
	    \sigma_{s}:=\sum_{y\in\cY}r(y)\ket{y}\bra{y}\otimes W_{s}(y). 
	\end{equation*}
\end{lemma}
    In Section \ref{Bcccodes} and Section \ref{TPCcodes}, we show that the above statements give us the desired codes for transmission of public and common messages. These statements generalize the coding results from \cite{mosonyi} to include pruned input distributions rather than distributions of n-fold product form.\\
    Finally, to obtain codes for transmission of confidential messages, we perform privacy amplification arguments on the public part of the codebook achieved from Lemma \ref{mosonyivariation} (cf.\cite{minglai}). To do so, we need the following inequality.
    \begin{theorem}[\cite{ahlswede02}, Theorem 19]\label{coveringlemma}
	Let $\mu > 0$, $ \epsilon \in (0, \frac{1}{2})$ be positive numbers and $X_1,\dots,X_L$ an independent and identically distributed family of positive semi-definite random matrices on $\mathbbm{C}^d$ such that the bounds $X \leq \mu\mathbbm{1}_{\bbmC^d}$ and $\mathbbm{E}X \geq \epsilon \mathbbm{1}_{\bbmC^d}$ apply. It holds 
	\begin{align*}
	\Pr\left(\left\|\frac{1}{L}\sum_{i=1}^L X_i - \mathbbm{E} X\right\|_{1}  \ > \ \epsilon\right) \ \leq \ 2 \cdot d \cdot \exp\left(-L \frac{\epsilon^3 }{2d\mu \ln 2}\right)
	\end{align*}
\end{theorem}

Equipped with these preliminary results, we prove the direct parts of the capacity theorems for BCC and TPC in the following two subsections. 
\subsection{ BCC codes}\label{Bcccodes}
In this section, we prove the following lemma.  
\begin{lemma}\label{directpartbcc}
	Let $\cW:=\{W_{s}\}_{s\in S}\subset CQ(\cX,\cH_{B}\otimes\cH_{E})$ be any compound cqq broadcast channel. It holds 
\begin{align*}
C_{BCC}[\cW]\supset\cl\left( \bigcup_{l=1}^{\infty}\bigcup_{p}\frac{1}{l}\hat{C}^{(1)}\big(\cW,p,l\big)\right),
\end{align*}
	where the second union is taken over all $p_{UYX}\in\cP(\cU\times\mathcal{Y}\times\cX^{l})$ such that random variable $U-Y-X$ form a Markov chain and alphabets $\cU$ and $\mathcal{Y}$ are finite. 
\end{lemma}
The main step towards proving Lemma \ref{directpartbcc}, is the following statement. 
\begin{lemma}[Broadcast channel with confidential messages ]\label{publiclyenhancedbcc}
	Let $\cW:=\{W_{s}\}_{s\in S}\subset CQ(\cX,\cH_{B}\otimes\cH_{E})$ be any compound cqq broadcast channel. For $p_{UYX}\in\cP(\cU\times\cY\times\cX)$ where $U - Y - X$ form a Markov chain and $\delta,\epsilon>0$, there exists $n_{0}\in\mathbb{N}$, such that for $n>n_{0}$, we find an $(n,M_{0,n},M_{c,n})$ BCC code $\cC=(E(\cdot|m),D_{B,m},D_{E,m_{0}})_{m=(m_{0},m_{c})\in [M_{0,n}]\times[M_{c,n}]}$ with
	\begin{enumerate}
		\item $\frac{1}{n}\log M_{0,n}\geq\inf_{s\in S}\min\left\{ I(U;B,\omega_{s},I(U;E,\omega_{s}\right\}-c\delta$,
		\item $\frac{1}{n}\log M_{c,n}\geq\inf_{s\in S}I(Y;B|U,\omega_{s})-\sup_{s\in S}I(Y;E|U,\omega_{s})-c\delta$\\
		with some constant $c>0$ and $\omega_{s}$ defined by (\ref{evaluationstate}). 
		\item $\inf_{s\in S}\frac{1}{|\mathbf{M}|}\sum_{\mathbf{m}\in\mathbf{M}}\sum_{\mathbf{x}\in\cX^{n}}E(\mathbf{x}|\mathbf{m})\tr[W_{B,s}^{\otimes n}(\mathbf{x})D_{B,\mathbf{m}}]\geq 1-\epsilon$
		\item $\inf_{s\in S}\frac{1}{|\mathbf{M}|}\sum_{\mathbf{m}\in\mathbf{M}}\sum_{\mathbf{x}\in\cX^{n}}E(\mathbf{x}|\mathbf{m})\tr[W_{E,s}^{\otimes n}(\mathbf{x})D_{E,m_{0}}]\geq 1-\epsilon$
		\item $\sup_{s\in S}I(M_{c};E|M_{0},\sigma_{s,n})\leq\epsilon$
	\end{enumerate}
	with state $\sigma_{s,n}$ defined by (\ref{securitystate}).
\end{lemma}
Applying standard double-blocking arguments on Lemma \ref{publiclyenhancedbcc}, will prove Lemma \ref{directpartbcc}. In the same vein as the coding steps taken in \cite{shaefer}, we prove Lemma \ref{publiclyenhancedbcc} in two steps. At first, we prove the following random coding result, that guarantees reliable decoding of common messages by Bob and Eve, and reliable decoding of public messages by Bob. Here, we do not concern ourselves with the security condition. In the next step, we apply privacy amplification arguments on the public part of the codebook, to achieve the desired confidential message transmission rate. 
\newpage

\begin{lemma}\label{lemmanonsecurebcc}
	Let $\cW:=\{W_{s}\}_{s\in S}\subset CQ(\cY,\cH_{B}\otimes\cH_{E})$ be any compound cqq broadcast channel and $\cU$ be a finite  alphabet. For any $\delta>0$, $q\in\cP(\cU), r(\cdot|u)\in\cP(\cY)$, $u\in\cU$ and large enough values of $n$,  the following exist.
	\begin{itemize}
	\item A family $(u_{m},D_{E,m})_{m\in[M_{0,n}]}$ of codes with $u_{m}\in T_{q,\delta}^{n}$ and $(D_{E,m})_{m\in[M_{0,n}]}\subset\cL(\cH_{E}^{\otimes n})$ a POVM.
	\item A map $y:(y_{ij})_{(i,j)\in[M_{0,n}]\times[M_{1,n}]}\mapsto(D_{B,ij}(y))_{(i,j)\in[M_{0,n}]\times[M_{1,n}]} $, such that $(D_{B,ij}(y))_{(i,j)\in[M_{0,n}]\times[M_{1,n}]}\in\cL(\cH_{B}^{\otimes n})$ is a POVM and for any family $Y=(Y_{ij})_{(i,j)\in[M_{0,n}]\times[M_{1,n}]}$ of random variables such that for each $m\in[M_{0,n}]$, $Y^{m}=(Y_{mj})_{j\in[M_{1,n}]}$ is  distributed i.i.d according to $r_{n}'(\cdot|u_{m})$, namely the pruned distribution of $r(\cdot|u)$ (see (\ref{pruned2})),  we have
\begin{equation*}
    \frac{1}{n}\log M_{0,n}\geq \inf_{s\in S}\min\left\{I(U;B,\omega_{s}), I(U;E,\omega_{s})\right\}-c\delta,
\end{equation*}
\begin{equation*}
    \frac{1}{n}\log M_{1,n}\geq \inf_{s\in S}I(Y;B|U,\omega_{s})-c\delta,
\end{equation*}
		\begin{equation*}
		    \mathbb{E}_{Y}\bigg[\inf_{s\in S}\frac{1}{M_{0,n}M_{1,n}}\sum_{(m,i)\in[M_{0,n}]\times[M_{1,n}]}\tr[W_{B,s}^{\otimes n}(Y_{mi})D_{B,mi}(Y)]\bigg]\geq 1-\epsilon_{n},
		\end{equation*}
		\begin{equation*}
		    \mathbb{E}_{Y}\bigg[\inf_{s\in S}\frac{1}{M_{0,n}M_{1,n}}\sum_{(m,i)\in[M_{0,n}]\times[M_{1,n}]}\tr[W_{E,s}^{\otimes n}(Y_{mi})D_{E,m}]\bigg]\geq 1-\epsilon_{n}
		\end{equation*}
		with $\epsilon_{n}\to 0$ exponentially, constant $c>0$ and $\omega_{s}=\sum_{u\in\cU}q(u)\ket{u}\bra{u}\otimes r(y|u)\ket{y}\bra{y}\otimes W_{s}(y)$.
	   \end{itemize}
\end{lemma}

\begin{proof}
We approximate $\{W_s\}_{s \in S}$ by a finite $\tau_n$-net $\{W_s\}_{s\in S_n} \subset \{W_s\}_{s \in S}$ with $\tau_n := 2^{-\tfrac{n\nu}{2}}$ with a constant positive number $\nu$ to be determined later. We choose the net small enough to have $\log|S_n| \ \leq \ 2\cdot |\cX| \cdot \dim(\cH_{B}\otimes\cH_{E})^2 (\log 6 +  n\nu /2)$ which is possible by Lemma \ref{netslemma}.
	For $\gamma\in\{B,E\}$ and $s\in S_{n}$, consider the effective channel $\hat{W}_{\gamma,s,n}:\cU^{n}\to \cS(\cH_{\gamma}^{\otimes n})$ defined by $\hat{W}_{\gamma,s}(\cdot):=\sum_{y\in\cY}r(y|\cdot)W_{\gamma,s}(y)$. Applying Lemma \ref{corollary} on the channel set $\{\hat{W}_{\gamma,s}\}_{s\in S_{n}}$ and probability distribution $q$, yields the existence of the random $(n,M_{0,n})$ code $\cC(U)$ with $U=(U_{1},\dots,U_{M_{0,n}})$, a sequence of i.i.d random variables distributed according to $q'_{n,\delta}$ and POVMs $(D_{\gamma,m}(U))_{m\in[M_{0,n}]}\subset\cL(\cH_{\gamma}^{\otimes n})$ such that 
	 \begin{align}\label{commonforbobbcc}
	 \mathbb{E}_{U}\big[\min_{s\in  S_{n}}\frac{1}{M_{0,n}}\sum_{m\in[M_{0,n}]}\tr(D_{\gamma,m}(U)\hat{W}_{\gamma,s}^{\otimes n}(U_{m}))\big]\geq 1-\epsilon_{0,n}.
	 \end{align}
	 with $\epsilon_{0,n}\to 0$ exponentially and 
	 \begin{equation*}
	     \frac{1}{n}\log M_{0,n}\geq\min_{s\in S_{n}}I(U;\gamma,\omega_{s})-c_{0}\delta. 
	 \end{equation*}
	Hence we have
	\begin{equation*}
	     \frac{1}{n}\log M_{0,n}\geq\min_{s\in S_{n}}\min\left\{I(U;B,\omega_{s}),I(U;E,\omega_{s})\right\}-c_{0}\delta. 
	 \end{equation*}
	Given (\ref{commonforbobbcc}), we can conclude the existence of one realization $(u_{1},\dots,u_{M_{0,n}})$ of random variable $U$, and POVMs $(D_{\gamma,m})_{m\in[M_{0,m}]}\in\cL(\cH_{\gamma}^{\otimes n})$, suitable for transmission of common messages, namely
	 \begin{equation}\label{bobrealizationbcc}
	 \min_{s\in S_{n}}\frac{1}{M_{0,n}}\sum_{m\in[M_{0,n}]}\tr(D_{\gamma,m}\hat{W}_{\gamma,s}^{\otimes n}(u_{m}))\geq 1-\epsilon_{0,n}.
	 \end{equation}
	 Before moving on to the private message, notice that for each $\mathbf{u}\in T_{q,\delta}^{n}$, using the abbreviation $T_{\delta}:=r^{\otimes n}(T_{r,\delta}(\mathbf{u}))$, we have
	 \begin{align}\label{prunedtoiid}
	 \|\hat{W}_{\gamma,s}^{\otimes n}(\mathbf{u})-\sum_{\mathbf{y}\in\cY^{n}}r_{n}'(\mathbf{y}|\mathbf{u})W_{\gamma,s}^{\otimes n}(\mathbf{y})\|_{1}&\leq \sum_{\mathbf{y}\in T_{r,\delta}(\mathbf{u})}r^{\otimes n}(\mathbf{y}|\mathbf{u})(\frac{1}{T_{\delta}}-1)\|W_{\gamma,s}^{\otimes n}(\mathbf{y})\|_{1}\nonumber\\&+\sum_{\mathbf{y}\in T^{c}_{r,\delta}(\mathbf{u})}r^{\otimes n}(\mathbf{y}|\mathbf{u})\|W_{\gamma,s}^{\otimes n}(\mathbf{y})\|_{1}\leq 2(1-T_{\delta})\leq 2\cdot2^{-n\delta}.
	 \end{align}
	 The upper bound above comes from the fact that $T_{\delta}$ approaches unity exponentially with $n$ (cf. \cite{korner}). Now we pursue with the private message, namely the one Bob has to decode while Eve may or may not. For each $u_{\hat{m}},\hat{m}\in[M_{0,n}]$ obtained above,  apply Lemma \ref{mosonyivariation} on $\{W_s\}_{S_n}$ and probability distribution $r(\cdot|u)$, $u\in\cU$.  on Lemma \ref{mosonyivariation}, we obtain the existence of a random code $\cC(Y^{u_{\hat{m}}})$ with $Y^{u_{\hat{m}}}=(Y_{\hat{m},1},\dots,Y_{\hat{m},M_{1,n}})$ and decoding operation
	 $(\Lambda_{m}(Y^{u_{\hat{m}}}))_{m\in[M_{1,n}]}$, such that $Y^{u_{\hat{m}}}$ is distributed according to $r'_{n,\delta}(\cdot|u_{\hat{m}})^{\otimes M_{1,n}}$ with
	 \begin{equation}\label{confidentialbobbcc}
	 \mathbb{E}_{Y^{u_{\hat{m}}}}\big[\inf_{s\in S_{n}}\frac{1}{M_{1,n}}\sum_{m\in[M_{1,n}]}\tr(\Lambda_{m}(Y^{u_{\hat{m}}})W^{\otimes n}_{B,s}(Y_{\hat{m},m}))\big]\geq 1-\epsilon_{1,n},
	 \end{equation}
	 and
	 \begin{equation*}
	     \frac{1}{n}\log M_{1,n}\geq I(Y;B|U,\omega_{s})-c_{1}\delta.
	 \end{equation*}
	We have
	 \begin{align}\label{bobdecodingcommonbcc}
	 \min_{s\in S_{n}}\frac{1}{M_{0,n}}&\sum_{\hat{m}\in[M_{0,n}]}\mathbb{E}_{Y^{u_{\hat{m}}}}\big[\frac{1}{M_{1,n}}\sum_{m\in[M_{1,n}]}\tr(D_{\gamma,\hat{m}}W_{s,\gamma}^{\otimes n}(Y_{\hat{m},m}))\big]\nonumber\\&=\min_{s\in S_{n}}\frac{1}{M_{0,n}}\sum_{\hat{m}\in[M_{0,n}]}(\tr(D_{\gamma,\hat{m}}\sum_{\mathbf{y}\in \cY^{n}}r'_{n,\delta}(\mathbf{y}|u_{\hat{m}})W_{s,\gamma}^{\otimes n}(\mathbf{y})))\nonumber\\&=\min_{s\in S_{n}}\frac{1}{M_{0,n}}\sum_{\hat{m}\in[M_{0,n}]}\tr(D_{\gamma,\hat{m}}\hat{W}_{\gamma,s}^{\otimes n}(u_{\hat{m}}))-2\cdot2^{-n\delta}\geq 1-\epsilon_{2,n}.
	 \end{align}
	 where in the first equality, we have calculated the expectation value given that for each $\hat{m}\in[M_{0,n}]$, $Pr(Y_{\hat{m},m}=\mathbf{y})=r_{n}'(\mathbf{y}|u_{\hat{m}}),\forall m\in[M_{0,n}]$, and in the last line, we have observed (\ref{prunedtoiid}) and inserted (\ref{bobrealizationbcc}), setting $\epsilon_{2,n}:=\epsilon_{0,n}+2\cdot2^{-n\delta}$. 
	  Consider the random decoding operation $(D_{B,\hat{m},m}(Y))_{(\hat{m},m)\in[M_{0,n}]\times[M_{1,n}]}$ defined for each message pair by $D_{B,\hat{m},m}(Y):=\sqrt{D_{B,\hat{m}}}\Lambda_{m}(Y^{u_{\hat{m}}})\sqrt{D_{B,\hat{m}}}$. We have 
	 \begin{align}\label{claim3bcc}
	 \mathbb{E}_{Y}\big[\frac{1}{|S_{n}|}\sum_{s\in S_{n}}\frac{1}{M_{0,n}}&\sum_{\hat{m}\in[M_{0,n}]}\frac{1}{M_{1,n}}\sum_{m\in[M_{1,n}]}\tr(D_{B,\hat{m},m}(Y)W_{B,s}^{\otimes n}(Y_{\hat{m},m}))\big]=\nonumber\\& \frac{1}{|S_{n}|}\sum_{s\in S_{n}}\frac{1}{M_{0,n}}\sum_{\hat{m}\in[M_{0,n}]}\mathbb{E}_{{Y^{u_{\hat{m}}}}}\big[\frac{1}{M_{1,n}}\sum_{m\in[M_{1,n}]}\tr(\sqrt{D_{B,\hat{m}}}\Lambda_{m}(Y^{u_{\hat{m}}})\sqrt{D_{B,\hat{m}}}W_{B,s}^{\otimes n}(Y_{\hat{m},m}))\big]\geq\nonumber\\& \frac{1}{|S_{n}|}\sum_{s\in S_{n}}\frac{1}{M_{0,n}}\sum_{\hat{m}\in[M_{0,n}]}\big(\mathbb{E}_{Y^{u_{\hat{m}}}}\big[\frac{1}{M_{1,n}}\sum_{m\in[M_{1,n}]}\tr(\Lambda_{m}(Y^{u_{\hat{m}}})W_{B,s}^{\otimes n}(Y_{\hat{m},m}))\big]-\nonumber\\&2\sqrt{1-\mathbb{E}\big[\frac{1}{M_{1,n}}\sum_{m\in[M_{1,n}]}\tr(D_{B,\hat{m}}W_{B,s}^{\otimes n}(Y_{\hat{m},m}))\big]} \big)\geq 1-\epsilon_{1,n}-2\sqrt{\epsilon_{2,n}},
	 \end{align}
	 where in the first inequality, we have used Lemma \ref{ineq}, and in the last line, we have inserted the lower bounds from (\ref{bobdecodingcommonbcc}) and (\ref{confidentialbobbcc}) and used concavity of the square root function.  Applying standard net approximation techniques used for example in proof of Lemma \ref{mosonyivariation}, we obtain the claim of the lemma. 
\end{proof}
At this point we can prove Lemma \ref{publiclyenhancedbcc}, by applying privacy amplification arguments (c.f \cite{minglai}) on the $M_{1}$ part of the messages obtained in Lemma \ref{lemmanonsecurebcc}. This is done by using equidistribution when inputting part of these messages to confuse the eavesdropper. The other part of $M_{1}$ will then be secure. 
	\begin{proof}[Proof of Lemma \ref{publiclyenhancedbcc}]
	Let $p_{UYX}(u,y,x)=p_{UY}(u,y)p_{X|Y}(x|y)$ and $p_{UY}(u,y)=q(u)r(y|u)$ $\forall (u,y,x)\in\cU\times\cY\times\cX$. We approximate $\{W_s\}_{s \in S}$ by a finite $\tau_n$-net $\{W_s\}_{S_n} \subset \{W_s\}_{s \in S}$ with $\tau_n := 2^{-\tfrac{n\nu}{2}}$ with a constant positive number $\nu$ to be determined later. We choose the net small enough to fulfill the cardinality bound 
		$\log|S_n| \ \leq \ 2\cdot |\cX| \cdot \dim(\cH_{B}\otimes\cH_{E})^2 (\log 6 +  n\nu /2)$ which is possible by Lemma \ref{netslemma}.
	Let $\delta>0$, $n\in\mathbb{N}$ and pruned probability distributions $q'_{n,\delta}, r_{n}'(\cdot|\mathbf{u})$  over $T_{q,\delta}^{n}$ and $T_{r,\delta}(\mathbf{u}),(\mathbf{u}\in\cU^{n})$ be given.
	Set
	\begin{equation}\label{m0}
	M_{0,n}=\lfloor 2^{n\big(\min_{s\in S_{n}}\min\left\{I(U;B,\omega_{s}),I(U;E,\omega_{s})\right\}-c\delta\big)}\rfloor,
	\end{equation}
	\begin{equation}\label{jn}
	J_{n}=\lfloor 2^{n\big(\min_{s\in S_{n}}I(Y;B|U,\omega_{s})-\max_{s\in S_{n}}I(Y;E|U,\omega_{s})-2\Delta(\delta)-c\delta\big)}\rfloor
	\end{equation}
	and
	\begin{equation}\label{Ln}
	L_{n}=\lceil 2^{n\max_{s\in S_{n}}I(Y;E|U,\omega_{s})+n\Delta(\delta)} \rceil.
	\end{equation}
	For the effective channel $\tilde{W}_{s}:\cY\to\cS(\cH_{B}\otimes\cH_{E})$ defined by $\tilde{W}_{s}(\cdot):=\sum_{x\in\cX}p_{X|Y}(x|\cdot)W_{s}(x), \forall s\in S$,  according to Lemma \ref{lemmanonsecurebcc}, there exists a family $(u_{m},D_{E,m})_{m\in[M_{0},n]}$ and a random family\\
 $\cC(Y)=(Y_{mjl},D_{B,mjl}(Y))_{(m,j,l)\in[M_{0,n}]\times[J_{n}]\times[L_{n}]}$, such that for events 
 \begin{equation*}
 \mathbf{A}:=\big\{\max_{s\in S_{n}}\frac{1}{M_{0,n}J_{n}L_{n}}\sum_{m,j,l}\tr(\tilde{W}_{B,s}^{\otimes n}(Y_{mjl})D^{c}_{B,mji}(Y))\geq \sqrt{\epsilon_{n}}\big\}
 \end{equation*}
 and
 \begin{equation}\label{droppedcondition}
 \mathbf{B}:=\big\{\max_{s\in S_{n}}\frac{1}{M_{0,n}J_{n}L_{n}}\sum_{m,j,l}\tr(\tilde{W}_{E,s}^{\otimes n}(Y_{mjl})D^{c}_{E,m})\geq  \sqrt{\epsilon_{n}}\big\}, 
 \end{equation}
 we have
	\begin{align}\label{reliability}
	\Pr[\mathbf{A}\cup \mathbf{B}]\leq  2\sqrt{\epsilon_{n}},
	\end{align} 
	where we have used the Markov inequality to obtain the above probability from the expectation value of the same event, and applied the union bound to get the probability of the complementary events (one with respect to $\tilde{W}_{B,s}$ and the other with respect to $\tilde{W}_{E,s}$). Here, $\epsilon_{n}$ goes to zero exponentially, given the appropriate choice of $\tau_{n}$, as evident in the proof of Lemma \ref{mosonyivariation}. We define the following quantities for each $s\in S_{n}$ and $\mathbf{u}\in T_{q,\delta}^{n}$. 
	\begin{align}
	Q_{s}^{\mathbf{u}}(\mathbf{y}):=\Pi_{\tilde{W}_{E,s},\mathbf{u},\delta}\Pi_{\tilde{W}_{E,s},\mathbf{u},\delta}(\mathbf{y})\tilde{W}_{E,s}^{\otimes n}(\mathbf{y})\Pi_{\tilde{W}_{E,s},\mathbf{u},\delta}(\mathbf{y})\Pi_{\tilde{W}_{E,s},\mathbf{u},\delta}
	\end{align}
	with quantities defied in Lemma \ref{typicalsets} and 
	\begin{align}
	\Theta_{s}^{\mathbf{u}}:=\sum_{y^{n\in\cY^{n}}}r'_{n,\delta}(\mathbf{y}|\mathbf{u})Q_{s}^{\mathbf{u}}(\mathbf{y}).
	\end{align}
	Given property 4 of Lemma \ref{typicalsets}, (\ref{totproj}) and Lemma \ref{gentle}, we have $\forall \mathbf{u}\in \cU^{n}, \mathbf{y}\in T_{r,\delta}(\mathbf{u})$ and $s\in S_{n}$
	\begin{align}\label{gm}
	\parallel\tilde{W}_{E,s}^{\otimes n}(\mathbf{y})-Q_{s}^{\mathbf{u}}(\mathbf{y})\parallel_{1}\leq\sqrt{2^{-n\Upsilon(\delta)+1}+\sqrt{2^{-n\Upsilon''(\delta)+2}}}:=\epsilon_{1,n}.
	\end{align}
	Clearly $\epsilon_{1,n}\to 0$ exponentially. Applying Theorem \ref{coveringlemma} with $\mathbb{C}^{d}$ the range space of projection $\Pi_{\tilde{W}_{E,s},\mathbf{u},\delta}$, by property 2 of Lemma \ref{typicalsets} we have 
	\begin{align}
	d\leq 2^{S(E|U,\omega_{s}^{\otimes n})+n\Delta(\delta)}
	\end{align}
	Furthermore, from the property $6$ of the projections introduced in Lemma\ref{typicalsets}, we have for all $\mathbf{u}\in T_{q,\delta}^{n}$
	\begin{align}\label{equation32}
	Q_{s}^{\mathbf{u}}(Y_{mjl})\leq 2^{-S(E|YU,\omega_{s}^{\otimes n})+n\Gamma'(\delta)}\mathbbm{1}_{\mathbb{C}^{d}}.
	\end{align}
	Let $n>2$. The hypotheses of Theorem \ref{coveringlemma} are therefore satisfied with $\epsilon=\epsilon_{0,n}:=2^{-n\Gamma'(\delta)/6}$ and $\mu=2^{-S(E|YU,\omega_{s}^{\otimes n})+n\Gamma'(\delta)}$. Since $u_{m}\in T_{q,\delta}^{n},\forall m\in[M_{0,n}]$, for the event
	\begin{equation*}
	\mathbf{C}_{s,m,j}:=\left\{\left\|\frac{1}{L_{n}}\sum_{l=1}^{L_{n}}Q_{s}^{u_{m}}(Y_{mjl})-\Theta_{s}^{u_{m}}\right\|_{1}>\epsilon_{0,n}\right\},
	\end{equation*}
	 we have
	\begin{align*}
	\Pr\big[\mathbf{C}_{s,m,j}\big]&\leq  2^{S(E|U,\omega_{s}^{\otimes n})+n\Delta(\delta)}\times\exp\big(-L_{n}\frac{\epsilon_{0,n}^{3}}{2\ln 2\cdot 2^{S(E|U,\omega_{s}^{\otimes n})-S(E|YU,\omega_{s}^{\otimes n})+n(\Delta(\delta)-\Gamma'(\delta))}}\big)\\&\leq  2^{n(\log\dim\cH_{E}+\Delta(\delta))}\times\exp\big(-L_{n}\frac{\epsilon_{0,n}^{3}}{2\ln 2\cdot 2^{I(Y;E|U,\omega_{s}^{\otimes n})+n(\Delta(\delta)-\Gamma'(\delta))}}\big).
	\end{align*}
	 Applying the union bound, for all $s\in S_{n},j\in J_{n},m\in[M_{0,n}]$ we have 
	\begin{align}\label{security}
	\Pr\big[\mathbf{C}:=\bigcup_{s,j,m}\mathbf{C}_{s,m,j} \big]\leq J_{n}M_{0,n}|S_n|2^{n(\log\dim\cH_{E}+\Delta(\delta))}\times\exp\big(-L_{n}\frac{\epsilon_{0,n}^{3}}{2\ln 2\cdot 2^{I(Y;E|U,\omega_{s}^{\otimes n})+n(\Delta(\delta)-\Gamma'(\delta))}}\big).
	\end{align}
	From (\ref{reliability}) and (\ref{security}), we have 
	\begin{align}\label{securityreliability}
	\Pr\big[\mathbf{C}\cup\mathbf{B}\cup\mathbf{A}\big]&\leq 2\sqrt{\epsilon_{n}}+ J_{n}M_{0,n}\times|S_n|2^{n(\log\dim\cH_{E}+\Delta(\delta))}\nonumber\\&\times\exp\big(-L_{n}\frac{\epsilon_{0,n}^{3}}{2\ln 2\cdot 2^{I(Y;E|U,\omega_{s}^{\otimes n})+n(\Delta(\delta)-\Gamma'(\delta))}}\big).
	\end{align}
 Finally, given (\ref{Ln}), we have
	\begin{equation}\label{ded}
	\exp\big(-L_{n}\frac{\epsilon_{0,n}^{3}}{2\ln 2\cdot 2^{I(Y;E|U,\omega_{s}^{\otimes n})+n(\Delta(\delta)-\Gamma'(\delta))}}\big)\leq \exp\big(-\frac{\epsilon_{0,n}^{3}2^{n\Gamma'(\delta)}}{2\ln 2\cdot }\big),
	\end{equation}
	which gives us a double exponential decay given that $\epsilon_{0,n}=2^{-n\Gamma'(\delta)/6}$. Inserting (\ref{ded}) in (\ref{securityreliability}), we conclude that we can find one realization $\{y_{mjl}\}_{(m,j,l)\in[M_{0,n}]\times[J_{n}]\times[L_{n}]}$ of $Y$, such that 
	\begin{equation}\label{reliabilitycrit}
	\min_{s\in S_{n}}\frac{1}{M_{0,n}J_{n}L_{n}}\sum_{m,j,l}\tr(\tilde{W}_{B,s}^{\otimes n}(y_{mjl})D_{B,mjl})\geq 1-\sqrt{\epsilon}_{n},
	\end{equation}
	\begin{equation}\label{reliabilitycriteve}
	\min_{s\in S_{n}}\frac{1}{M_{0,n}J_{n}L_{n}}\sum_{m,j,l}\tr(\tilde{W}_{E,s}^{\otimes n}(y_{mjl})D_{E,m})\geq 1-\sqrt{\epsilon}_{n}
	\end{equation}
	and 
	\begin{equation}\label{securitycrit}
	\max_{s\in S_{n}}\max_{m,j}\left\|\frac{1}{L_{n}}\sum_{l=1}^{L_{n}}Q_{s}^{u_{m}}(y_{mjl})-\Theta_{s}^{u_{m}}\right\|_{1}\leq \epsilon_{0,n}.
	\end{equation}
	Consider the stochastic encoder $E(\cdot|m,j):=\frac{1}{L_{n}}\sum_{l\in [L_{n}]}p^{n}_{X|Y}(\cdot|y_{mjl})$ and POVM $(D_{B,mj}:=\sum_{l\in [L_{n}]}D_{mjl})_{m,j}$. Therefore with $M_{c,n}=J_{n}$, we have 
	\begin{align}\label{averageb}
	\inf_{s\in S}\frac{1}{M_{0,n}M_{c,n}}\sum_{m_{0}\in[M_{0}],m_{c}\in[M_{c,n}]}&\tr(\sum_{\mathbf{x}\in\cX^{n}}E(\mathbf{x}|m_{0},m_{c})W_{B,s}^{\otimes n}(\mathbf{x})D_{B,m_{0},m_{c}})\nonumber\\&=\frac{1}{M_{0,n}M_{c,n}}\sum_{m_{0},m_{c}}\tr(\frac{1}{L_{n}}\sum_{\mathbf{x}\in\cX^{n}}\sum_{l\in [L_{n}]}p^{n}_{X|Y}(\cdot|y_{m_{0}m_{c}l})W_{B,s}^{\otimes n}(\mathbf{x})\sum_{l'\in [L_{n}]}D_{B,m_{0}m_{c}l'})\nonumber\\&\geq 1-\sqrt{\epsilon}_{n}-2n\tau_{n},
	\end{align}
	where in the last line we have inserted the bound from (\ref{reliabilitycrit}) and observed that the error due to $\{W_{s}\}_{s\in S_{n}}$ can only be $2n\tau_{n}$ less than the error due to $\cW$. By the same line of reasoning we have
	\begin{align}\label{averageeve}
	\inf_{s\in S}\frac{1}{M_{0,n}M_{c,n}}\sum_{m_{0}\in[M_{0}],m_{c}\in[M_{c,n}]}&\tr(\sum_{\mathbf{x}\in\cX^{n}}E(\mathbf{x}|m_{0},m_{c})W_{E,s}^{\otimes n}(\mathbf{x})D_{E,m_{0}})\nonumber\\&=\frac{1}{M_{0,n}M_{c,n}}\sum_{m_{0},m_{c}}\tr(\frac{1}{L_{n}}\sum_{\mathbf{x}\in\cX^{n}}\sum_{l\in [L_{n}]}p^{n}_{X|Y}(\cdot|y_{m_{0}m_{c}l})W_{E,s}^{\otimes n}(\mathbf{x})\sum_{l'\in [L_{n}]}D_{E,m_{0}})\nonumber\\&\geq 1-\sqrt{\epsilon}_{n}-2n\tau_{n}. 
	\end{align}
	The 5th claim in the statement of the lemma related to the security criterion requires upper bounding $\sup_{s\in S}I(M_{c};E|M_{0},\sigma_{s,n})$ for all $s\in S_{n}$, that is done in the following. First we observe that for all $s\in S$
	\begin{align}\label{cond:mutual:security}
	    	I(M_{c};E|M_{0},\sigma_{s,n})= \frac{1}{M_{0,n}}\sum_{m_{0}\in[M_{0,n}]}I(M_{c};E,\sigma_{s,n}^{m_{0}}),
	\end{align}
	with 
	\begin{equation*}
	    \sigma_{s,n}^{m_{0}}:=\frac{1}{M_{c,n}}\sum_{m_{c}\in[M_{c,n}]}\otimes\sum_{\mathbf{x}\in\cX^{n}}E(\mathbf{x}|m_{0},m_{c})W^{\otimes n}(\mathbf{x}).
	\end{equation*}
	We continue upper-bounding the mutual information on the right hands side of (\ref{cond:mutual:security}) for each $m_{0}\in[M_{0,n}]$. We note that for all $s\in S_{n}$
	\begin{align}\label{defin:cond:security}
	    I(M_{c};E,\sigma_{s,n}^{m_{0}})&=S\bigg(\frac{1}{M_{c,n}}\sum_{m_{c}\in[M_{c,n}]}\sum_{\mathbf{x}\in\cX^{n}}E(\mathbf{x}|m_{0},m_{c})W_{E,s}^{\otimes n}(\mathbf{x})\bigg)\nonumber\\&-\frac{1}{M_{c,n}}\sum_{m_{c}\in[M_{c,n}]}S\bigg(\sum_{\mathbf{x}\in\cX^{n}}E(\mathbf{x}|m_{0},m_{c})W_{E,s}^{\otimes n}(\mathbf{x})\bigg)=S\bigg(\frac{1}{M_{c,n}L_{n}}\sum_{j\in[M_{c,n}],l\in[L_{n}]}\tilde{W}_{E,s}^{\otimes n}(y_{m_{0}jl})\bigg)\nonumber\\&-\frac{1}{M_{c,n}}\sum_{m_{c}\in[M_{c,n}]}S\bigg(\sum_{l\in[L_{n}]}\tilde{W}_{E,s}^{\otimes n}(y_{m_{0}jl})\bigg)
	\end{align}
	Notice that, given (\ref{gm}) and (\ref{securitycrit}) and the triangle inequality we have for all $s\in S_{n}$
	\begin{equation}\label{close:to:uniform}
	    \parallel\sum_{l\in[L_{n}]}\tilde{W}_{E,s}^{\otimes n}(y_{m_{0}jl})-\Theta_{s}^{u_{m_{0}}}\parallel_{1}\leq\epsilon_{0,n}+\epsilon_{1,n}.
	\end{equation}
	Applying Lemma \ref{fannesap} with $\delta=\epsilon_{0,n}+\epsilon_{1,n}$, given (\ref{close:to:uniform}) and (\ref{defin:cond:security}) we obtain
	\begin{align}
	    I(M_{c};E,\sigma_{s,n}^{m_{0}})\leq 2\left(n(\epsilon_{0,n}+\epsilon_{1,n})\log\dim(\cH_{E})+h(\epsilon_{0,n}+\epsilon_{1,n})\right).
	\end{align}
 Inserting this into (\ref{cond:mutual:security}), we obtain the same upper bound on the conditional mutual information quantity on the left hand side for all $s\in S_{n}$. Given properties of the $\tau$-net (Lemma \ref{netslemma}), applying Lemma \ref{cond:mutual:inf:cont} with $\delta=2n\tau_{n}$ we obtain
	 \begin{align}\label{netsecurity}
	 \sup_{s\in S}I(M_{c};E|M_{0},\sigma_{s,n})&\leq\max_{s\in S_{n}}I(M_{c};E|M_{0},\sigma_{s,n})+2\left(2n^{2}\tau_{n}\log\dim(\cH_{B})+(1+2n\tau_{n})h(2n\tau_{n}/1+2n\tau_{n})\right)\nonumber\\&\leq 2\left(n(\epsilon_{0,n}+\epsilon_{1,n})\log\dim(\cH_{E})+h(\epsilon_{0,n}+\epsilon_{1,n})\right)\nonumber\\&+2\left(2n^{2}\tau_{n}\log\dim(\cH_{B})+(1+2n\tau_{n})h(2n\tau_{n}/1+2n\tau_{n})\right).
	 \end{align}
Given the upper bound on $S_{n}$, choosing $\nu=\frac{1}{8n|\cX|\dim(\cH_{B}\otimes\cH_{E})}\log\epsilon_{n}$, we obtain exponential decay of the right hand sides of (\ref{averageb}) and (\ref{averageeve}). Also, with this value of $\tau_{n}$ and choosing large enough values of $n$, (\ref{netsecurity}) gives us the 5th claim of the statement.\\
	 \end{proof}
\begin{proof}[Proof of Lemma \ref{directpartbcc}]
	According to Lemma \ref{publiclyenhancedbcc}, 
		\begin{equation*}
		(R_{0},R_{c})\in\bigcup_{p}\hat{C}^{(1)}\big(\mathcal{J},p,1\big)
		\end{equation*}
		implies $(R_{1},R_{c})\in C_{BCC}(\mathcal{J})$. Using standard double-blocking and time sharing arguments, for each $l\in\mathbb{N}$, 
		\begin{equation*}
		(R_{0},R_{c})\in \cl\left(\bigcup_{l=1}^{\infty}\bigcup_{p}\frac{1}{l}\hat{C}^{(1)}\big(\mathcal{J},p,l\big)\right),
		\end{equation*}
	implies	$(R_{0},R_{c})\in C_{BCC}(\mathcal{J})$.
\end{proof}
Here, in order to construct private codes for the broadcast channel, we first generated suitable random message transmission codes for the broadcast channel without imposing privacy constraints (Lemma \ref{lemmanonsecurebcc}). This was done by establishing suitable bounds for random universal "superposition codes". Subsequent application of a covering principle these codes where transformed to fulfill the security criterion in Lemma \ref{publiclyenhancedbcc}. 
Beside technical obstacles to construct superposition codes for cq broadcast channels which are robust regarding uncertainty of the channel state, the approach is rather traditional and even dates back to classical information theory (see e.g. \cite{korner} for a general discussion, the classical counterpart to our considerations can be found in \cite{shaefer}). \newline 

\subsection{TPC codes}\label{TPCcodes}
In this section, we prove the following lemma.  
\begin{lemma}\label{directparttpc}
	Let $\cW:=\{W_{s}\}_{s\in S}\subset CQ(\cX,\cH_{B}\otimes\cH_{E})$ be any compound cqq broadcast channel. It holds 
\begin{align*}
C_{TPC}[\cW]\supset\cl\left(\bigcup_{l=1}^{\infty}\bigcup_{p}\frac{1}{l}C^{(1)}\big(\cW,p,l\big)\right),
\end{align*}
	where the second union is taken over all $p_{VYX}\in\cP(\cV\times\mathcal{Y}\times\cX^{l})$ such that random variable $V-Y-X$ form a Markov chain and alphabets $\cV$ and $\mathcal{Y}$ are finite. 
\end{lemma}
The main step towards proving Lemma \ref{directparttpc}, is the following statement. 
\begin{lemma}[Broadcast channel with confidential messages ]\label{publiclyenhancedtpc}
	Let $\cW:=\{W_{s}\}_{s\in S}\subset CQ(\cX,\cH_{B}\otimes\cH_{E})$ be any compound cqq broadcast channel. For $p_{VYX}\in\cP(\cV\times\cY\times\cX)$ where $V-Y-X$ form a Markov chain and $\delta,\epsilon>0$, there exists $n_{0}\in\mathbb{N}$, such that for $n>n_{0}$, we find an $(n,M_{1,n},M_{c,n})$ TPC code $\{E(\cdot|m),D_{B,m},D_{E,m_{1}}\}_{m=(m_{1},m_{c})\in [M_{1,n}]\times[M_{c,n}]}$ with
	\begin{enumerate}
		\item $\frac{1}{n}\log M_{1,n}\geq\inf_{s\in S} I(V;B,\omega_{s})-c\delta$,
		\item $\frac{1}{n}\log M_{c,n}\geq\inf_{s\in S}I(Y;B|V,\omega_{s})-\sup_{s\in S}I(Y;E|V,\omega_{s})-c\delta$\\
		with some constants $c>0$ and $\omega_{s}$ defined by (\ref{evaluationstate}). 
		\item $\inf_{s\in S}\frac{1}{|\mathbf{M}|}\sum_{\mathbf{m}\in\mathbf{M}}\sum_{\mathbf{x}\in\cX^{n}}E(\mathbf{x}|\mathbf{m})\tr[W_{B,s}^{\otimes n}(\mathbf{x})D_{B,\mathbf{m}}]\geq 1-\epsilon$
		\item $\sup_{s\in S}I(M_{c};E|M_{1},\sigma_{s,n})\leq\epsilon$
	\end{enumerate}
	with state $\sigma_{s,n}$ defined by (\ref{securitystate}).
\end{lemma}
Applying standard double-blocking arguments on Lemma \ref{publiclyenhancedtpc}, will prove Lemma \ref{directparttpc}. We prove Lemma \ref{publiclyenhancedtpc} in two steps. At first, we prove the following random coding result, that guarantees reliable decoding of public messages by Bob and Eve, and reliable decoding of public messages by Bob. In the next step, we apply privacy amplification arguments on the public part of the codebook, to achieve the desired confidential message transmission rate. 
\begin{lemma}\label{lemmanonsecuretpc}
Let $\cW:=\{W_{s}\}_{s\in S}\subset CQ(\cY,\cH_{B}\otimes\cH_{E})$ be any compound cqq broadcast channel and $\cV$ be a finite  alphabet. For any $\delta>0$, $q\in\cP(\cV), r(\cdot|v)\in\cP(\cY)$, $v\in\cV$ and large enough values of $n$, the following exist.
	\begin{itemize}
	\item A family $(v_{m})_{m\in[M_{2,n}]}$ of words with $v_{m}\in T_{q,\delta}^{n}$.
	\item A map $y:(y_{ij})_{(i,j)\in[M_{1,n}]\times[M_{2,n}]}\mapsto(D_{B,ij}(y))_{(i,j)\in[M_{1,n}]\times[M_{2,n}]} $, such that $(D_{B,ij}(y))_{(i,j)\in[M_{1,n}]\times[M_{2,n}]}\in\cL(\cH_{B}^{\otimes n})$ is a POVM and for any family $Y=(Y_{ij})_{(i,j)\in[M_{1,n}]\times[M_{2,n}]}$ of random variables such that for each $m\in[M_{1,n}]$, $Y^{m}=(Y_{mj})_{j\in[M_{2,n}]}$ is  distributed i.i.d according to $r'(\cdot|v_{m})$  we have
\begin{equation*}
    \frac{1}{n}\log M_{1,n}\geq \inf_{s\in S}I(V;B,\omega_{s})-c\delta,
\end{equation*}
\begin{equation*}
    \frac{1}{n}\log M_{2,n}\geq \inf_{s\in S}I(Y;B|V,\omega_{s})-c\delta,
\end{equation*}
		\begin{equation*}
		    \mathbb{E}_{Y}\bigg[\inf_{s\in S}\frac{1}{M_{1,n}M_{2,n}}\sum_{(m,i)\in[M_{1,n}]\times[M_{2,n}]}\tr(W_{B,s}^{\otimes n}(Y_{mi})D_{B,mi}(Y)\bigg]\geq 1-\epsilon_{n},
		\end{equation*}
		with $\epsilon_{n}\to 0$ exponentially, constant $c>0$ and $\omega_{s}=\sum_{v\in\cV}q(u)\ket{v}\bra{v}\otimes r(y|v)\ket{y}\bra{y}\otimes W_{s}(y)$.
	   \end{itemize}
\end{lemma}

\begin{proof}
The proof is done by following exactly the lines in proof of Lemma\ref{lemmanonsecurebcc}, except that here $\gamma=\{B\}$.
\end{proof}
\begin{proof}[Proof of Lemma \ref{publiclyenhancedtpc}]
The proof follows by applying the privacy amplification arguments in the proof of Lemma \ref{publiclyenhancedbcc}, on $[M_{2,n}]$ part of the messages in Lemma \ref{lemmanonsecuretpc}. It is clear that here, we only consider upper bounding the probability of events corresponding to events $\mathbf{A}$ and $\mathbf{C}$ in the proof of that Lemma \ref{publiclyenhancedbcc}, and drop (\ref{droppedcondition}). 
\end{proof}
\begin{proof}[Proof of Lemma \ref{directparttpc}]
	According to Lemma \ref{publiclyenhancedtpc}, 
		\begin{equation*}
		(R_{1},R_{c})\in\bigcup_{p}C^{(1)}\big(\mathcal{J},p,1\big)
		\end{equation*}
		implies $(R_{1},R_{c})\in C_{TPC}(\mathcal{J})$. Using standard double-blocking and time sharing arguments, for each $l\in\mathbb{N}$, 
		\begin{equation*}
		(R_{1},R_{c})\in \cl\left(\bigcup_{l=1}^{\infty}\bigcup_{p}\frac{1}{l}C^{(1)}\big(\mathcal{J},p,l\big)\right),
		\end{equation*}
	implies	$(R_{1},R_{c})\in C_{TPC}(\mathcal{J})$.
\end{proof}
\section{Outer bounds for the capacity regions} \label{sect:outer_bounds}
In this section, we consider the "converse" bounds stated in Theorem \ref{bcctheorem} and Theorem \ref{tpctheorem}. The arguments of proof turn out to be fairly standard. Therefore, we restrict ourselves to providing proof details regarding the outer bound to the BCC capacity regions from Theorem \ref{bcctheorem}.
\begin{proposition}\label{prop:bcc_converse}
Let $\cW := \{W_s\}_{s \in S}$, $W_s: \cX \rightarrow \cS(\cH_B \otimes \cH_E)$, $(s \in S)$ be a set of cqq channels. It holds 
\begin{align*}
C_{BCC}[\cW]&\subset\cl \left(\bigcup_{l=1}^{\infty}\bigcup_{p}\frac{1}{l}\hat{C}^{(1)}(\cW,p,l)\right).
\end{align*}
	The second union is taken over all $p_{UYX}\in\cP(\cU\times\mathcal{Y}\times\cX^{l})$ such that random variable $U-Y-X$ form a Markov chain and alphabets $\cU$ and $\mathcal{Y}$ are finite. 
\end{proposition} 
\begin{proof}
	Let $(\cC_n)_{n \in \bbmN}$ be a sequence of $(n,M_{1,n},M_{c, n})$ BCC codes for $\cW$ such that with a sequence $e_n \rightarrow 0$, $(n \rightarrow \infty)$ for all $s \in S$ $\overline{e}_B(\cC_n, W_s^{\otimes n })$, $\overline{e}_E(\cC_n, W_s^{\otimes n })$ and $I(M_{c,n}; E^n| M_{0,n}, \sigma_{s,n})$
		%\begin{align*}
	%	\overline{e}_B(\cC_n, W_s^{\otimes n }), \hspace{.5cm} \overline{e}_E(\cC_n, W_s^{\otimes n }) \hspace{.5cm} \text{and} \hspace{.5cm} I(M_{c,n}; E^n| M_{0,n}, \sigma_{s,n})
	%\end{align*}
	are simultaneously upper-bounded by $e_n$. While we fix the blocklength for a moment (and suppress the index $n$), we consider for each $s \in S$ the quadruple $(M_0,M_c,M^{(s)}_0,M^{(s)}_c)$ of random variables, where $M^{(s)}_0,M^{(s)}_c$ belong to the common and confidential messages decoded by $B$ after transmission with $W_s^{\otimes n}$. Note, that $\Pr((M_0,M_c) \neq (M^{(s)}_0, M^{(s)}_c)) \leq \epsilon_n$ is true by assumption. It holds
	\begin{align}
		\log M_{0} 
		 \ = \ H(M_0)
		 \ = \ I(M_0 ; M^{(s)}_0) + H(M_0|M^{(s)}_0)
		 \ \leq \ I(M_0; B^n, \sigma_{s,n}) + \epsilon_n \cdot \log M_0. \label{prop:bcc_converse_common_1}
	\end{align} 
	The second of the above equalities is the chain rule for the mutual information. The last inequality stems from application of Fano's lemma and the Holevo bound. A similar calculation for the second receiver leads us to the inequality
	\begin{align}
		\log M_{0} \ \leq I(M_0; E^n, \sigma_{s,n}) + \epsilon_n \cdot \log M_0. \label{prop:bcc_converse_common_2}
	\end{align}
	Maximizing over all $s \in S$ in (\ref{prop:bcc_converse_common_1}) and (\ref{prop:bcc_converse_common_2}) and combining the resulting inequalities gives the bound 
	\begin{align*} 
		\log M_0 \ \leq \ \min \left\{\sup_{s \in S} I(M_0; B^n, \sigma_{s,n}), \ \sup_{s \in S}I(M_0; E^n, \sigma_{s,n}) \right\} \ + \ \epsilon_n \log M_0. 
	\end{align*}
	In order to derive a bound on $M_c$, we notice the inequality 
	\begin{align}
		\log M_0 \cdot M_c 
		& \ \leq \ I(M_0M_c; B^n, \sigma_{s,n}) + \epsilon_n \cdot \log M_0 M_c. \label{prop:bcc_converse_prod_bound}
    \end{align}	
	The chain rule for the quantum mutual information implies 
	\begin{align*} 
		I(M_0M_c; B^n, \sigma_{s,n}) - \log M_0 \ \leq I(M_0M_c; B^n, \sigma_{s,n}) - I(M_0;B^n, \sigma_{s,n}) \ = \ I(M_c;B^n|M_0, \sigma_{s,n}).
	\end{align*} 
	Combining the above inequality with (\ref{prop:bcc_converse_prod_bound}) and 
	rearranging terms give us the inequality
	\begin{align*}
		\log M_c \ 
		& \ \leq \ I(M_c; B^n|M_0, \sigma_{s,n}) + \epsilon_n \cdot \log M_0 M_c.
	\end{align*}
	Maximizing both	sides of the inequality and adding the nonnegative term $\epsilon_n - \sup_{s \in S} I(M_c;E^n|M_0, \sigma_{s,n})$ to the right hand side of the result, we obtain
	\begin{align}
		\log M_c \ 
		\leq \ \sup_{s \in S} I(M_c; B^n|M_0, \sigma_{s,n}) \ - \ \sup_{s \in S} I(M_c; E^n|M_0, \sigma_{s,n})  \ + \ \epsilon_n (\log M_0 \cdot M_c + 1). \label{prop:bcc_converse_confid}
	\end{align}
	Let $\delta > 0$ be arbitrary and $n_0$ large enough for $\epsilon_n (\log M_0 \cdot M_c) 
	\leq \delta$ to hold. It is clear, for each $n> n_0$, $(\frac{1}{n} \log M_{0,n}, \frac{1}{n} \log M_{c,n})$ is contained in
	\begin{align}
	  \bigcup_{l > n_0} \frac{1}{n} \bigcup_{p} \hat{C}^{(1)}(\cW, p, n)_\delta \ \subset \ \left[\bigcup_{l=1}^{\infty}\bigcup_{p}\frac{1}{l}\hat{C}^{(1)}(\cW,p,l) \right]_\delta,
	\end{align}
where $A_{\delta}$ is the $\delta$-blowup of $A$ for each $\delta>0$ and $A\in\mathbb{R}^{+}_{0}\times\mathbb{R}^{+}_{0}$, i.e
\begin{equation*}
A_{\delta}:=\{y\in\mathbb{R}^{+}_{0}\times\mathbb{R}^{+}_{0}:\exists x\in A : \parallel x-y \parallel\leq\delta\}.
\end{equation*}
Since $\delta$ was an arbitrary positive number, we are done.
\end{proof}
\begin{proposition}\label{prop:tpc_converse}
	Let $\cW := \{W_s\}_{s \in S}$, $W_s: \cX \rightarrow \cS(\cH_B \otimes \cH_E)$, $(s \in S)$ be a set of cqq channels. It holds 
	\begin{align*}
	C_{TPC}[\cW]& \ \subset \ \cl \left(\bigcup_{l=1}^{\infty}\bigcup_{p}\frac{1}{l}C^{(1)}(\cW,p,l)\right).
	\end{align*}
\end{proposition} 
	The second union is taken over all $p_{VYX}\in\cP(\cV\times\mathcal{Y}\times\cX^{l})$ such that random variable $V-Y-X$ form a Markov chain and alphabets $\cV$ and $\mathcal{Y}$ are finite. 
\begin{proof}
	The proof can be conducted following exactly the same strategy as in the proof of Proposition \ref{prop:bcc_converse}, and therefore is left to the reader. The only modification is, that there is no need for $E$ to decode the message $M_1$ (opposed to the case of $M_0$ in the proof of Proposition \ref{prop:bcc_converse}).
	This leads to the bound 
	\begin{align*}
	\log M_1 \ \leq \ \sup_{s \in S} I(M_0; B^n, \sigma_{s,n}) +  \epsilon_n \log M_1. 
	\end{align*}
	on the number public messages in the code.
\end{proof}
\section {BCC and TPC capacities of compound quantum broadcast channels} \label{sect:full_quantum}
In this section we extend our results to the "full quantum" setting where the receivers input quantum systems to the channels, i.e. the transition maps of the channels are c.p.t.p. maps instead of cq channels. Since the message transmission tasks we aim to perform are after all of a classical nature, the corresponding coding theorems can be proven applying the results from earlier chapters. \newline 
Explicitely we apply the results of the preceding sections to derive codes for full quantum broadcast channels. For the remainder of this section, we fix an arbitrary set $\cJ := \{\cN_s\}_{s \in S}$,\ where 
\begin{align*}
\cN_s: \cL(\cH_A) \rightarrow \cL(\cH_B \otimes \cH_E)
\end{align*}
is a c.p.t.p. map for each $s \in S$. Traditionally, the c.p.t.p. map $\cN_s$ is assumed to be an isometric channel, namely a Stinespring isometry to a given channel connecting $A$ and $B$. This way of defining the channel is fairly justified, since it naturally equips $E$ with the strongest abilities when attacking the confidential transmission goals of the remaining parties. However, dropping this assumption on the channel does not complicate any subsequent arguments. \newline 
In what follows, we consider the BCC scenario. Corresponding considerations regarding the TPC scenario are easily extrapolated and are hence left to the reader.  
\begin{definition}[BCC codes]
	An \emph{$(n, M_0, M_c)$ BCC code} for $\cJ$ for channels in $\cC(\cH_A, \cH_B \otimes \cH_E)$ is a family $\cC = (V(m), D_{B,m}, D_{E,m_0})_{m \in \mathbf{M}}$ with $\mathbf{M} := [M_0]\times [M_c]$, where $(D_{B,m})_{m \in \mathbf{M}}$ and $(D_{E,m_0})_{m_0 \in [M_0]}$ are POVMs on $\cH_B^{\otimes n}$ resp. $\cH_E^{\otimes n}$ and $V(m)$ is a state on $\cH_A^{\otimes n}$ for each $m$. 
\end{definition}
The average transmission errors for the receivers $B$, and $E$ with channel $\cN: \cL(\cH_A) \rightarrow \cL(\cH_B \otimes \cH_E)$ and $(n,M_0,M_c)$-code $\cC$ are defined by
\begin{align*}
	\overline{e}_{B}(\cC, \cN^{\otimes n}) \ 
	:= \ \frac{1}{|\mathbf{M}|}\sum_{m \in \mathbf{M}} \ \tr D_{B,m}^c \cN^{\otimes n}(V(m)), \hspace{.5cm} \text{and} \hspace{.5cm}
	\overline{e}_{E}(\cC, \cN^{\otimes n}) \ 
	:= \ \frac{1}{|\mathbf{M}|}\sum_{m \in \mathbf{M}} \ \tr D_{E,m_0}^c \cN^{\otimes n}(V(m)).   
\end{align*}
By replacing the code and errors the definitions of achievable rate pairs can be directly guessed from Definition \ref{bccratepairdef} (the notational ambiguity should cause no misunderstandings since the set $\cJ$ determines whether the classical-quantum or quantum broadcast channel scenario are considered.) We denote the corresponding \emph{BCC capacity region} by $C_{BCC}[\cJ]$ . We moreover define $\hat{C}^{(1)}(\cJ, p, l, (\rho_y)_{y \in \cY})$ the set of all points in $\bbmR^2$ which fulfill the inequalities
\begin{align*}
  0  \ &\leq \ R_0  \leq  \underset{s \in S}{\inf} \ \min \ \left\{I(U;B,\omega_s), \ I(U;E,\omega_s)\right\} \hspace{2cm} \text{and} \\
  0  \ &\leq \ R_c  \leq  \underset{s \in S}{\inf}  \ I(Y;B|U, \omega_s) \ - \ \underset{s \in S}{\inf} \ I(Y;E|U, \omega_s)
  \end{align*} 
 where we understand the entropic quantities above as being evaluated on the ccq state
 \begin{align*}
 \omega_s \ := \ \omega(\cN_s, p, l) \ := \ \sum_{u \in \cU, y \in \cY} P_{UY}(u,y) \cdot \ket{u, y}\bra{u,y} \otimes \cN_{s}^{\otimes l}(\rho_{y})  
 \end{align*}
	for each $s \in S$. 
\begin{theorem}
It holds 
\begin{align*}
C_{BCC}[\cJ] \ = \ \cl\left(\bigcup_{l=1}^{\infty}\bigcup_{p} \frac{1}{l} \hat{C}^{(1)}(\cJ,p, l)\right)
\end{align*}
	The second union is taken over all $p_{UYX}\in\cP(\cU\times\mathcal{Y}\times\cX^{l})$ such that random variable $U-Y-X$ distributed accordingly, form a Markov chain and alphabets $\cU$ and $\mathcal{Y}$ are finite. 
\end{theorem}
\begin{proof}
	The proof of achievability is easily performed by referring to the corresponding result for ccq broadcast channels. Namely, if we fix $l \in \bbmN$, probability distributions $P_U$ and $P_{Y|U}$ and a family $(\rho_y)_{y \in \cY}$ of quantum states on $\cH_A^{\otimes n}$ we have
	\begin{align*}
		\omega(\cN_s, p, l, (\rho_y)_{y \in \cY}) \ = \ \sum_{u \in \cU}\sum_{y \in \cY} P_U(u)\cdot P_{Y|U}(y|u) \ket{u,y}\bra{u,y} \otimes \cN_s^{\otimes l}(\rho_y) \ = \ \omega(\tilde{V}_s,p,1)
	\end{align*}
	with an effective cqq channel with signals $\widetilde{V}_s(y) \ := \ \cN_s^{\otimes l}(\rho_y)$, $(y \in \cY)$. As a consequence, $\frac{1}{l} \hat{C}^{(1)}(\cJ, p, l, (\rho_y)_{y \in \cY}) = \frac{1}{l} \hat{C}^{(1)}(\{\widetilde{V}_s\}_{s \in S}, p, 1)$. We know from Theorem \ref{bcctheorem}, that each point on the r.h.s. of the preceding inequality is achievable. To prove the converse, we assume, that $\cC_n := (D_{B,m}, D_{E,m_0}, V(m))_{m \in \mathbf{M}}$ is an $(n, M_0, M_c)$-code with
	\begin{align*}
		\overline{e}_B(\cC_n, \cN_s^{\otimes n}),, \  \overline{e}_E(\cC_n,\cN_s^{n}), \hspace{1cm} \text{and} \hspace{1cm} I(M_c;E|M_0, \sigma_{s,n})
	\end{align*} 
	are simultaneously bounded by $\epsilon_n \in (0,1)$. Note, that the mutual information quantity above is evaluated on the code state 
	\begin{align*}
	\sigma_{s,n} := \frac{1}{\mathbf{M}} \sum_{m \in \mathbf{M}} \ket{m}\bra{m} \otimes \cN_s^{\otimes n}(V(m)).
	\end{align*}
	Using the above bounds and repeating the corresponding steps from the proof of Proposition \ref{prop:bcc_converse}, we obtain the inequalities 
	\begin{align*} 
	\log M_0  \leq  \min \left\{\sup_{s \in S} I(M_0; B^n, \sigma_{s,n}), \ \sup_{s \in S}I(M_0; E^n, \sigma_{s,n}) \right\} \ + \ \epsilon_n \log M_0. 
	\end{align*}
	and $\log M_c  \leq  I(M_c; B^n|M_0, \sigma_{s,n}) + \epsilon_n \cdot \log M_0 M_c $
	The remaining steps directly carry over from the cqq converse.
\end{proof}

\section{Concluding remarks and future work}\label{concludingremarks}
To construct private codes for the broadcast channel, we first generated suitable random message transmission codes for the broadcast channel without imposing privacy constraints (Lemma \ref{lemmanonsecurebcc}). This was done by establishing suitable bounds for random universal "superposition codes". With subsequent application of a covering principle, these codes were transformed to fulfill the security criterion in Lemma \ref{publiclyenhancedbcc}.\\
 As a possible alternative technique to generate such codes, we mention the rather recent "position  decoding" and "convex split" techniques \cite{anshu17,anshu19}. This approach proved to be powerful yet elegant and was successfully applied to determine "one-shot capacities" or "second order rates" in several scenarios. However, these techniques need still to be further developed, to also be suitable when dealing with channel uncertainties as in the scenarios considered in the present paper. A partial result in that directions is \cite{anshu19a}, where near-optimal universal codes for entanglement assisted message transmission over compound quantum channels with finitely many channel states are constructed. Recently, convex split and position-decoding have been applied in \cite{wilde19} to determine the second-order capacity of a cqq compound wiretap channel under the restriction, that the channel state does not vary for the legitimate receiver. For establishing this result, only the "convex split" part has to be universal, while "position- decoding" is applied on a channel with fixed state. As a future research goal, it is desirable to close the gap and establish a fully universal version of these protocol steps. \newline
 As mentioned in the introduction, a strong converse cannot be established for the message transmission capacity of the compound cq channel under average error criterion, even when considering $|S|=2$. When considering a fixed non-vanishing upper bound on the average of decoding error, calculation of capacity for the compound channel is further problematic as there are examples where the optimistic definition of the $\epsilon$-capacity yields a strictly larger number than the one yielded by its pessimistic definition (see \cite{BSP} Remark 13). This implies that in general there is no second rate capacity theorem possible. The implications of these negative statements are highly interesting in practice, as channel coding in all existing communication systems (such as wireless cellular and WiMax systems), is done given a fixed error probability. It is therefore important to design channel codes corresponding to $\epsilon$-capacity of the compound channel, that is in general larger than its message transmission capacity.  \\
 When considering the one-shot approach (\cite{anshu17,anshu19,salek}) as an alternative to proving capacity results derived here, one must take certain consequences into account. In this approach, one tries to obtain lower and upper-bounds for the $\epsilon$-capacity, and then consider the limit $\epsilon\to 0$ of these bounds. For the compound channel however, the capacity is in general strictly smaller than the $\epsilon$-capacity and hence, it is not clear how these bounds will help, as a lower bound on the $\epsilon$-capacity is not a priori a lower bound on the capacity of the channel. Furthermore, there are some additional highly interesting properties of the $\epsilon$-capacity and the capacity, even when one considers \textit{finite} compound channels ($|S|<\infty$):
 \begin{itemize}
 	\item The capacity of the finite compound channel is, as a function of the computable compound channel, a Turing computable function. This is no longer true about infinite compound channels (see \cite{icassp2020}).
 	\item The $\epsilon$-capacity of the finite compound channel, as a function of $\epsilon$, is not Banach-Mazur computable, which in turn means that it is not Turing computable either, as the latter condition is a stronger one on computability than former. 
 \end{itemize}
These results have of course an impact on the effectiveness of the one-shot approach to achieving capacity results in classical and quantum information theories \cite{icassp2020}.\\

A direction for future work given the results derived here, is considering a three dimensional capacity region, establishing a trade-off between the ability of the quantum channel in transmitting common, public and confidential messages under assumptions of the compound channel model. One must pay attention to the operational difference between public messages (belonging to the set $[M_{1,n}]$) and those used for equivocation by Alice (belonging to the set $[L_{n}]$). \\
Another direction for future work given the results derived here, is considering the arbitrarily varying quantum channel (AVQC) model for the broadcast channel with confidential messages. Given that in all instances, our error and security requirements, achieve exponential rates of decay, it is perceivable that using the well known robustification and elimination techniques developed in \cite{ahls}, capacity results including dichotomy statements can be made for the AVQC model. 

\section*{Acknowledgments}
H. Boche is partly supported by the Deutsche Forschungsgemeinschaft (DFG, German Research Foundation), under Germany's Excellence Strategy EXC-2111-390814868. G. Jan\ss en is partly supported by the Gottfried Wilhelm Leibniz program of the DFG under Grant BO 1734/20-1 and by the national research initiative on quantum technologies of Bundesministerium f\"ur Bildung und Forschung (German Federal Ministry of Education and Research), within the
project Q-Link.X 16KIS0858. S. Saeedinaeeni is partly supported by the national research initiative on quantum technologies of Bundesministerium f\"ur Bildung und Forschung (German Federal Ministry of Education and Research), within the project QuaDiQua 16KIS0948.

\appendix
\numberwithin{equation}{section}
\section{Auxiliary results}
We use the following statements in the text. 
\begin{lemma}[Tender operator \cite{wilde13}]\label{gentle}
Let $\rho\in\cS(\cH)$ and $T\in\cL(\cH)$ be a positive operator with $T\leq\mathbbm{1}_{\cH}$ and $1-\tr(\rho T)\leq\epsilon\leq 1$. Then we have 
\begin{equation*}
    \parallel\rho-\sqrt{T}\rho\sqrt{T}\parallel_{1}\leq\sqrt{2\epsilon}. 
\end{equation*}
\end{lemma}
\begin{lemma}\label{continuity}
	Let $\{W_{s}:\cX\to\cS(\cH)\}_{s\in S}$ be a set of cq channels and let $p\in\cP(\cX)$. Then
	\begin{equation*}
	\lim_{\alpha\to 1}\inf_{s\in S}\chi_{\alpha}(W_{s},p)=\inf_{s\in S}\chi(p,W_{s}).
	\end{equation*}
\end{lemma}
\begin{lemma}\label{ineq}
	Let $p,q\in\cL(\cH)$, $0\leq p,q\leq \mathbbm{1}_{\cH}$ and $\tau\in\cS(\cH)$. It holds
	\begin{equation*}
	\tr(\tau pqp)\geq \tr(\tau q)-2\sqrt{\tr(\tau(\mathbbm{1}-p))}.
	\end{equation*}
\end{lemma}
\begin{lemma}[cf.\cite{fannes73}]\label{fannesap}
	For any two states $\rho$ and $\sigma$ on Hilbert space $\cH$, let $\delta=\parallel\rho-\sigma\parallel_{1}$ and $\dim(\cH)=d$. Then 
\begin{equation}
	    |S(\rho)-S(\sigma)|\leq\delta\log(d-1)+h(\delta)
	\end{equation}
	hold, with $h(x)=-(1-x)\log(1-x) - x\log x$, for $x\in(0,1]$ and $h(0)=0$, the binary entropy.
\end{lemma}
\begin{lemma}\label{cond:mutual:inf:cont}[cf. \cite{shirokov}, Corollary 2]
Let $\rho,\sigma\in\cS(\cH_{A}\otimes\cH_{B}\otimes\cH_{C})$, with $\parallel\rho-\sigma\parallel_{1}=\delta$ and $\dim(\cH_{B})=d$. It holds
\begin{equation*}
    |I(A;B|C,\rho)-I(A;B|C,\sigma)|\leq 2\left(\delta\log d + (1+\delta)h(\frac{\delta}{1+\delta})\right),
\end{equation*}
wit $h$, the binary entropy, as defined in the previous lemma.
\end{lemma}
\section{Universal classical-quantum superposition coding} \label{mosonyiapp}
In this appendix, we establish a random coding construction of superposition codes for classical-quantum channels which are a major ingredient for the achievability proofs in Section \ref{sect:coding_for_broadcast}. In particular a detailed proof of Lemma \ref{mosonyivariation} is provided. \newline
Over the years several code constructions for message transmission over compound cq channels have been established (see \cite{bjelakovic09, hayashi09, datta10, mosonyi}). 
The arguments we invoke below for proving Lemma \ref{mosonyivariation} rely heavily on the techniques Mosonyi's work \cite{mosonyi}. Therein properties of the quantum Renyi Divergences and the closely related "sandwich Renyi divergences" are used to derive universal random coding results for classical-quantum channels. Below we further elaborate on that approach and extend it by suitable superposition codes. \newline  
To facilitate connecting the discussion below with the arguments in \cite{mosonyi} we introduce some notation from there. For a probability distribution $p\in\cP(\cY)$ and a cq channel $W:\cY\to\cS(\cH)$, we define quantum states
\begin{align*}
W(p):=\sum_{y\in\cX}p(y) \cdot W(y), \hspace{1cm}
\mathbb{W}(p):=\sum_{y\in\cX}p(y)\ket{y}\bra{y}\otimes W(y), \text{and} \hspace{2.0cm} \hat{p}:=\sum_{y\in\cX}p(y)\ket{y}\bra{y}.
\end{align*}
For each pair of non-zero positive semi-definite operators $\rho,\sigma$ and every $\alpha\in(0,1)$
we define
\begin{align*}
Q_{\alpha}(\rho||\sigma):=\tr(\rho^{\alpha}\sigma^{1-\alpha}),
 \hspace{1cm} \text{and} \hspace{1cm}
D_{\alpha}(\rho||\sigma):=\frac{1}{1-\alpha}\log\tr(\rho^{\alpha}\sigma^{1-\alpha})
\end{align*}
from 
\begin{align}
    \chi_\alpha(p,W) \ := \ \underset{\sigma \in \cS(\cH)}{\inf} D_\alpha(\mathbb{W}(p)||\hat{p} \otimes \sigma)
\end{align}
derives. It is known, that the limit $\alpha \rightarrow 1$ of the above quantity exists and equals the Holevo quantity $\chi(p,W)$. Translating to the notation in the statement of Lemma 8, we notice, that $\chi(p,W) \ = \ I(Y;B)$ holds. 
\begin{lemma}\label{lemma:appendix_mosonyi_continuity}
Let $\cW$ be a set of cq channels each mapping $\cY$ to $\cS(\cK)$, $q$ a probability distribution on $\cX$ and $r_x$ a probability distribution on $\cY$ for each $x \in \cX$. It holds
    \begin{align*}
    \underset{\alpha \nearrow 1}{\lim} \ \underset{V \in \cW}{\inf} \ \sum_{x \in \cX} q(x) \cdot \chi_\alpha(r_x, V) \ 
    = \ \underset{V \in \cW}{\inf} \ \sum_{x \in \cX} q(x) \cdot \chi(r_x, V) 
    \end{align*}
\end{lemma}
The above statement slightly generalizes that of Lemma 3.13 in \cite{mosonyi} (regarding the limit from below). The proof is by a similar argument. We include a proof for the readers convenience. 
\begin{proof}
Set $f(\alpha, V) := \sum_{x \in \cX} \ q(x) \cdot \chi_{\alpha}(r_x, V)$    
for each $\alpha \in (0,1)$ and cq channel $V$. It holds
\begin{align*}
    \underset{\alpha \nearrow 1}{\lim} \ \underset{V \in \cW}{\inf} \ f(\alpha, V) \
    &\overset{\text{(a)}}{=} \ \underset{\alpha \nearrow 1}{\lim} \ \underset{V \in \overline{\cW}}{\min} \ f(\alpha, V) \\
    &\overset{\text{(b)}}{=} \ \underset{\alpha \in (0,1)}{\sup} \ \underset{V \in \overline{\cW}}{\min} \ f(\alpha, V)  \\
    &\overset{\text{(c)}}{=} \ \underset{V \in \overline{\cW}}{\min} \ \underset{\alpha \in (0,1)}{\sup} \ f(\alpha, V)  \\
    &\overset{\text{(d)}}{=} \ \underset{V \in \overline{\cW}}{\min} \ \underset{\alpha \nearrow 1}{\lim} \ f(\alpha, V)  \\
    &= \ \underset{V \in \cW}{\inf} \ \sum_{x \in \cX} q(x) \cdot \chi(V, r_x) 
    \end{align*}
    The equality in (a) holds by continuity of  $f(\alpha, \cdot)$ for each $\alpha \in (0,1)$ (the closure might be performed in any norm, for example the $\|\cdot\|_{CQ}$), (b) is justified, because the argument of the limit is monotonously increasing on (0,1). The min-max exchange in (c) is an application of Lemma 2.3 from \cite{mosonyi}, (d) by monotonicity of $f$ in $\alpha$, and in (e) the limit $\alpha \nearrow 1$ is performed according to Lemma B.3 in \cite{mosonyi2}.
    \end{proof}
The starting point for our proof of Lemma \ref{mosonyivariation} is the generic random coding bound from Hayashi and Nagaoka \cite{hayashi03} we state below. 
\begin{lemma}[\cite{hayashi03}, cf. \cite{mosonyi}, Lemma 4.15]\label{Lemma4}
	Let $V:\cY\to\cS(\cK)$ be a cq channel, $M\in\mathbb{N}$ and $p\in\cP(\cY)$. There exists a map $(y_{1},\dots,y_{M})\mapsto(\Lambda_{1}(y),\dots,\Lambda_{M}(y))$, such that $(\Lambda_{m}(y))_{m\in[M]}\subset\cL(\cK)$ is a POVM, and, given $Y^{M}:=(Y_{1},\dots,Y_{M})$ of independent random variables each with distribution $p$, for each $\forall \alpha\in(0,1)$, the bound
	\begin{align*}
	\mathbb{E}_{Y^{M}}\big[\frac{1}{M}\sum_{m\in[M]}\tr(W(Y_{m})\Lambda_{m}(Y^{M})^c)\big] \ 
	\leq \ 8 \cdot M^{1-\alpha} \cdot Q_{\alpha}\big(\mathbb{W}(p)||\hat{p}\otimes W(q)\big)
	\end{align*}
holds. 
\end{lemma}
\begin{proof}[Proof of Lemma \ref{mosonyivariation}]
	Fix $n \in \bbmN$ and an $n$-word $\mathbf{x} \in T_{q,\delta}^n$ which we assume to be of type $\lambda$ (i.e. $\mathbf{x}\in T_{\lambda}^{n}$). We approximate $\{W_s\}_{s \in S}$ by a finite $\tau_n$-net $\{W_s\}_{s \in S_n} \subset \{W_s\}_{s \in S}$ with $\tau_n := 2^{-\tfrac{n\nu}{2}}$ with a constant positive number $\nu$ to be determined later. We choose the net small enough to fulfill the cardinality bound 
	%\begin{align*}
	$\log|S_n| \ \leq \ 2\cdot |\cX| \cdot d^2 (\log 6 +  n\nu /2)$
	%\end{align*}
	which is possible by Lemma \ref{netslemma}. We introduce abbreviations $d := \dim \cK_B$, $r_{\mathbf{x}}(\cdot):=r^{\otimes n}(\cdot|\mathbf{x})$ and $r'_{\mathbf{x}}(\cdot):=r'_{n,\delta}(\cdot|\mathbf{x})$ for each $\mathbf{x} \in \cX^n$. 
	Applying Lemma \ref{Lemma4} on the cq channel $\overline{W}_{n} := \frac{1}{|S_n|}\sum_{s\in S_n}W_{s}^{\otimes n}$ with $p := r'_{\mathbf{x}}$, and
		\begin{align}
		M := \lfloor \exp(n ( \inf_{s \in S} I(Y;B|X,\sigma_s) - \delta\cdot|\cX|\log d - \nu)) \rfloor, 
	\end{align}
	we know that choosing a codewords $Y_1,\dots,Y_M$ i.i.d. according to $r_{\mathbf{x}}'$ each, allows us to bound the expectation by
\begin{align}
	\mathbb{E}_{Y^{M}}\left[\frac{1}{M}
	\sum_{m\in[M]} \tr(\overline{W}_n(Y_{m})\Lambda_{m}(Y^{M})^c)\right] \
	\leq \ 8\cdot M^{1-\alpha} \cdot Q_{\alpha}\left(\frac{1}{|S_n|}\sum_{s\in S_n}\mathbb{W}_{s}^{\otimes n}(r'_{\mathbf{x}})||\hat{r'}_{\mathbf{x}}\otimes \overline{W}_n(r'_{\mathbf{x}})\right)
	\label{lemma4app}
	\end{align}
for each $\alpha\in(0,1)$. By linearity of the trace and the expectation, the above inequality implies
\begin{align}
\mathbb{E}_{Y^{M}}\left[\min_{s \in S_n} \frac{1}{M}
\sum_{m\in[M]} \tr(W_s^{\otimes n}(Y_{m})\Lambda_{m}(Y^{M})^c)\right] \
\leq \ 8\cdot |S_n| \cdot M^{1-\alpha} \cdot Q_{\alpha}\left(\frac{1}{|S_n|}\sum_{s\in S_n}\mathbb{W}_{s}^{\otimes n}(r'_{\mathbf{x}})||\hat{r'}_{\mathbf{x}}\otimes \overline{W}_n(r'_{\mathbf{x}})\right). \label{mosonyivariation_prf_1}
\end{align}
The left hand side of the above inequality can be identified as the expected average error of a random code. We proceed to further bound the Function $Q_{\alpha}$ on the right hand side. We have
\begin{align}
Q_{\alpha}\bigg(\frac{1}{|S_n|}\sum_{s\in S_n}\mathbb{W}_{s}^{\otimes n}(r'_{\mathbf{x}})||\hat{r'}_{\mathbf{x}}\otimes \overline{W_{n}}(r'_{\mathbf{x}})\bigg)\leq \frac{1}{r_{\mathbf{x}}(T_{r,\delta}(\mathbf{x}))^{2-\alpha}}Q_{\alpha}\bigg(\frac{1}{|S_n|}\sum_{s\in S_n}\mathbb{W}_{s}^{\otimes n}(r_{\mathbf{x}})||\hat{r}_{\mathbf{x}}\otimes \overline{W_{n}}(r_{\mathbf{x}})\bigg). \label{mosonyivariation_prf_2}
	\end{align}
In (\ref{mosonyivariation_prf_2}) we have used definition of the pruned distribution, and observed operator monotonicity of the function $f(x)=x^{\alpha}$ for $\alpha\in[0,1]$ (cf. \cite{bhatia97}, Theorem 5.1.9). Following the arguments in proof of Lemma 4.16 \cite{mosonyi} we obtain
\begin{align}
Q_{\alpha}	
	\big(\frac{1}{|S_n|}\sum_{s\in S_n}\mathbb{W}_{s}^{\otimes n}(r_{\mathbf{x}})||\hat{r}_{\mathbf{x}}\otimes \overline{W_{n}}(q^{\otimes n})\big)
&	\ \leq \ \frac{1}{|S_n|^{\alpha}}\sum_{s\in S_n}\exp\big((\alpha-1)\cdot \alpha \cdot \chi_{\alpha}(r_{\mathbf{x}},W_{s}^{\otimes n})\big)\cdot d^{n(\alpha-1)^{2}} \nonumber \\ 
&	\ \leq \  \exp\left((\alpha-1)\cdot \alpha \cdot \min_{s\in S_n}\chi_{\alpha}(r_{\mathbf{x}}, W_{s}^{\otimes n}) + n (\alpha -1)^2 \log d + \log|S_n|\right)  \label{mosonyivariation_prf_3}
\end{align}
In order to further estimate the error exponent above, we note that for each $s \in S$
 \begin{align*}
     \chi_{\alpha}(W_{s}^{\otimes n},r_{\mathbf{x}}) \ 
     = \  \sum_{x\in\cX} \lambda(x) \cdot \chi_{\alpha}(W_{s},r(\cdot|x))  
     \ \geq \ \sum_{x\in\cX}q(x) \cdot \chi_{\alpha}(W_{s},r(\cdot|x))-\delta\cdot|\cX|\log d.
\end{align*}
In the above, we have used $|\lambda(x)-q(x)|\leq\delta$ and $\chi_{\alpha}(W_{s},r(\cdot|x))\leq \log d$. By Lemma \ref{lemma:appendix_mosonyi_continuity}, choosing $\alpha$ close enough to one allows us to bound 
\begin{align}
\alpha\inf_{s\in S}\sum_{x\in\cX}q(x) \cdot \chi_{\alpha}(r(\cdot|x),W_{s}) \
\geq \ \inf_{s\in S}\sum_{x\in\cX} q(x) \cdot \chi(r(\cdot|x),W_{s}) - \delta\cdot|\cX|\log d
= \ \inf_{s\in S} \ I(Y;B|X,\sigma_s) - \delta\cdot|\cX|\log d.
\end{align}
where we introduced the notation from the statement of Lemma 8 in the second line. Note, that with our 
choice of $M$, we have
\begin{align}
  \alpha\inf_{s\in S}\sum_{x\in\cX}q(x) \cdot \chi_{\alpha}(r(\cdot|x),W_{s}) - \frac{1}{n} \log M 
  \ \geq \ \nu > 0 \label{mosonyivariation_prf_4}
\end{align}
Combining the estimates from (\ref{mosonyivariation_prf_1}) -  (\ref{mosonyivariation_prf_4}) and subsequent upper-bounding the right hand side of  (\ref{lemma4app}), we achieve the bound 
\begin{align*}
 \mathbb{E}_{Y^{M}}\left[\min_{s \in S_n} \frac{1}{M} \sum_{m\in[M]} \tr(W_s^{\otimes n}(Y_{m})\Lambda_{m}(Y^{M})^c)\right] 
 & \ \leq 16 \cdot \exp \left((\alpha -1)\cdot n (\nu + (\alpha -1)\log d + 2 |\cX| d^2 [\frac{\log6}{n} + \frac{\nu}{2}] ) \right)  \\ 
 & \leq \ 2^{-n\nu/4}
 \end{align*} 
 Where the last inequality above holds for a fixed choice of $\alpha$ close enough to one and large enough $n$. By the property of the $\tau_n$ net and linearity of trace and expectation, we can conclude 
 \begin{align}
 \mathbb{E}_{Y^{M}}\left[\inf_{s \in S} \frac{1}{M} \sum_{m\in[M]} \tr(W_s^{\otimes n}(Y_{m})\Lambda_{m}(Y^{M})^c)\right] 
 & \leq \ 2^{-n\nu/4} +  n \cdot 2^{-n\nu}. 
 \end{align}
 We are done.
\end{proof}
\section{Net approximation of arbitrary channel sets}\label{nets}
In this section, we introduce the concept of finite nets for approximation of arbitrary channel sets used in proof of the direct parts of our capacity theorems.
\begin{definition}\label{netsdef}
For $\cW\subset CQ(\cX,\cH)$ and $\tau>0$, a $\tau$-net is a finite set $\cW_{\tau}:=\{W_{i}\}_{i\in S_{\tau}}\subset CQ(\cX,\cH)$, with property that for every $W\in\cW$, there exists and index $i\in [S_{n}]$ such that 
\begin{equation}
    \parallel W- W_{i}\parallel_{CQ}<\tau.
\end{equation}
\end{definition}
The existence of such $\tau$-net does not readily guarantee that $\cW_{\tau}\subset\cW$. The following lemma gives the existence of a good $\tau$-net contained in the given channel set. 
\begin{lemma}\label{netslemma}(cf. \cite{zeroerror} Lemma 6)
Let $\cW:=\{W_{i}\}_{i\in S}\subset CQ(\cX,\cH)$ and $\tau\in(0,1/e)$. There exists a set $\cW_{\tau}:=\{W_{i}\}_{i\in S_{\tau}}\subset\cW$ with  such that 
\begin{enumerate}
    \item $|S_{\tau}|<(\frac{6}{\tau})^{2|\cX|\dim(\cH)^{2}}$,
\item given any $n\in\mathbb{N}$, for every $i\in S$, there exists $i'\in S_{n}$ such that 
\begin{equation}
    \parallel W_{i}^{\otimes n}(\mathbf{x})- W_{i'}^{\otimes n}(\mathbf{x})\parallel_{1}\leq 2n\tau, \text{   } (\forall \mathbf{x}\in\cX^{n}).
\end{equation}
\end{enumerate}
\end{lemma}
	
    \end{document}